\documentclass[12pt,a4paper]{article}
\usepackage{cite} 
\usepackage[margin=2.5cm]{geometry}

\title{M-polynomial Based Mathematical Formulation of the Hyperbolic Sombor Index}
\author{Jayjit Barman and Shibsankar Das\thanks{Corresponding author.}\\
Department of Mathematics, Institute of Science,\\ Banaras Hindu University, Varanasi-221005, Uttar Pradesh, India.\\
Email: jayjit44@bhu.ac.in, shib.cgt@gmail.com\\
	}
\usepackage{fancyhdr}
\usepackage{lineno,hyperref}
\modulolinenumbers[5]

\usepackage{amssymb}
\usepackage{amsmath}
\usepackage{url}
\usepackage{enumerate}
\usepackage{fancyhdr}
\usepackage{amsthm}
\usepackage{cleveref}
\usepackage{xcolor}
\usepackage{float}
\usepackage{caption}
\usepackage{subfigure}
\usepackage{wrapfig}
\usepackage{multirow}
\usepackage{rotating}

\usepackage{graphicx}\graphicspath{{../}{../image/}{Figure/}{3_Word/}}
\newtheorem{definition}{Definition}
\newtheorem{theorem}{Theorem}

\newtheorem{corollary}[theorem]{Corollary}

\allowdisplaybreaks

\usepackage{datetime}
\date{\today ~\ampmtime}

\begin{document}
\maketitle

\begin{abstract}
The numerical values extracted from a graph that indicates its topology are called topological indices. A contemporary and efficient method is to compute a graph's topological indices using the graph polynomial that corresponds to it. This method of identifying degree-based topological indices involves the use of the M-polynomial.
Very recently, in $2025$, the hyperbolic Sombor index (\textit{HSO}) was proposed and shows its chemical applicability for octane isomers and the structure sensitivity and abruptness for octane, nonane, and decane isomers, respectively.
In this work, we establish the closed derivation formula for the above-mentioned index of a graph based on its M-polynomial.
Additionally, we use our proposed derivation formula to calculate the hyperbolic Sombor index of a few standard graphs and chemical families. Moreover, we provide the numerical and graphical representations for the M-polynomial and the computed \textit{HSO} index of the chemical families. 
\end{abstract}

\section{Introduction}
Chemical graph theory (CGT) is a broad area of mathematical chemistry that uses graph theory to model chemical molecules mathematically. Numerous degree-based topological indices have been presented and extensively studied in the mathematics and chemical literature concerning chemical graph theory.
Topological indices are numerical values in CGT that are derived from the graphical representation of the molecular structure utilizing mathematical invariants. They are crucial in the investigation of quantitative structure-property relationships (\textit{QSPR}) and quantitative structure-activity relationships (\textit{QSAR})~\cite{trinajstic1992}.

Let $G(V(G),E(G))$ be an undirected, simple and connected graph with vertex set $V(G)$ and  edge set $E(G)$. The degree $d(v)$ of a vertex $v$ in a graph $G$ is the number of edges that are connected to it. In graph $G$, an edge is represented by $e = uv$ or $vu$, where $u$ and $v$ are adjacent vertices~\cite{west2000}.

The hyperbolic Sombor index of a graph $G$ was first presented by Barman et al.~\cite{barman2025} in $2025$, and its mathematical representation is given as
$$\textit{HSO}(G)=\sum_{uv\in E(G)}\frac{\sqrt{d^2(u)+d^2(v)}}{d(u)},~\text{where}~0 <d(u)\le d(v).$$
Another representation of this index is as follows
	$$\textit{HSO}(G)=\sum_{uv\in E(G)}\frac{\sqrt{d^2(u)+d^2(v)}}{\min\{d(u),d(v)\}}.$$

Topological indices are typically determined based on their mathematical definitions; however, calculating them for large and complex graphs or networks is often a challenging and time-consuming task. To address this problem, various graph-invariant-based polynomials have been proposed in the literature. These polynomials are employed to derive topological indices by applying the required operators at a specific point. Several types of topological indices can now be calculated through the polynomial approach rather than computing them individually~\cite{sdas2022b, sdas2022g, deutsch2015, hosoya1988}.
The Hosoya polynomial~\cite{hosoya1988}, Tutte polynomial~\cite{kauffman1989}, Clar covering polynomial~\cite{zhang1996}, matching polynomial~\cite{farrell1979}, Schultz polynomial~\cite{hassani2013}, M-polynomial~\cite{deutsch2015}, neighborhood M-polynomial~\cite{mondal2021b,sdas2023}, and others are some examples of graphic polynomials that are used for obtaining the value for different classes of topological indices.

Deutsch and Klav\v{z}ar introduced the idea of the M-polynomial~\cite{deutsch2015} in $2015$ to study degree-based topological indices. The M-polynomial allows scientists to effectively examine molecular structures by storing the structural properties of molecules in a polynomial form. The M-polynomial approach to computing the degree-based topological indices can become laborious for large or densely connected graphs or networks. Determining M-polynomials for different graph families has been the subject of much research due to the complexity of computing graph polynomials and extracting topological indices as graphs get larger~\cite{sdas2020a,sdas2022a,sdas2022e}. We now provide a few fundamental definitions that are crucial to understanding this topic.
\begin{definition}~\cite{deutsch2015}
	\label{Def:MP}
	The M-polynomial of a graph $G$ is defined as
	\begin{equation}
		M(G;x,y)=\sum_{\delta \le i \le j \le \Delta} m_{ij}x^iy^j,
	\end{equation}
where $\delta= \min\{d(v):v\in V(G)\}$, $\Delta= \max\{d(v):v\in V(G)\}$ and $m_{ij}$ is the number of edges $uv \in E(G)$ such that $d(u)=i, d(v)=j,$ where $i,j \ge 1.$
\end{definition}

A degree-based topological index for a simple connected graph $G$, as specified in~\cite{deng2011} is denoted as $I(G)$ and stated as
\begin{equation}
	\label{Eq:DBTI}
	I(G)=\sum_{uv \in E(G)}f(d(u),d(v)),
\end{equation}
where $f(d(u),d(v))$ is a function of $d(u)$ and $d(v)$, which related to the corresponding topological indices. It is possible to rewrite the Equation~\ref{Eq:DBTI} by counting the edges which has same end degrees of vertices as
\begin{equation}
	\label{Eq:ADBTI}
	I(G)=\sum_{\delta \le i \le j \le \Delta} m_{ij}f(i,j)
\end{equation}
\begin{table}[htb!]
	\caption{There are some operators~\cite{afzal2020, basavanagoud2019, basavanagoud2020, sdas2022a, sdas2020a, sdas2021, sdas2022c, sdas2022e, deng2011, deutsch2015, hosoya1988, hussain2021} used to formulate the M-polynomial of the hyperbolic Sombor index:}
	\label{Table:Ope}
	\renewcommand*{\arraystretch}{1.5}
	\tiny
	\resizebox{\textwidth}{!}
	{
		\begin{tabular}{lc}
			\hline $D^{1/2}_{x}(f(x,y))= \sqrt{x\frac{\partial(f(x,y))}{\partial x}}\cdot \sqrt{f(x,y)}$, & $S_{x}(f(x,y))= \int_{0}^{x} \frac{f(t,y)}{t}dt$,\\
			$P_{x}(f(x^s,y^t))= f(x^{s^2},y^t)$, & $P_{y}(f(x^s,y^t))= f(x^s,y^{t^2})~\text{where}~s,t \in \mathbb{N} \cup \{0\}$,\\
			$J(f(x,y))=f(x,x)$.\\ \hline
		\end{tabular}
	}
\end{table}
\section*{Methodology}
In Section~\ref{Sec:CDF}, we present and illustrate a closed derivation formula for the previously mentioned hyperbolic Sombor index (\textit{HSO}), which may be used to calculate the index over the M-polynomial of a graph.
We determine the hyperbolic Sombor index values for several well-known graphs (such as a complete bipartite graph, $r$-regular graph, complete graph, cycle graph, star graph and path graph) and chemical families (such as boron icosahedral $\alpha$ sheet $\mathcal{B}_{\alpha}(a,b)$, dendrimers, jagged-rectangle benzenoid system~$B_{m,n}$, polycyclic aromatic hydrocarbons~$\textit{PAH}_{n}$, V-Phynelenic nanotubes~$\textit{VPHX}[m,n]$, V-Phynelenic nanotori~$\textit{VPHY}[m,n]$, porous graphene $\textit{PG}[p,q]$, tadpole graph~$T(n,m)$ and polyphenylenes $P[s,t]$) in Section~\ref{Sec:4} by utilizing their corresponding M-polynomial.
Additionally, we illustrate the surface representation of the M-polynomial of the specified chemical families and their hyperbolic Sombor index for various parameters related to them using MATLAB R2019a software. We also tabulate the values of the \textit{HSO} index of the chemical families.
Lastly, we conclude in Section~\ref{Sec:Con}.
\section{Closed Derivation Formula over the M-polynomial for the Hyperbolic Sombor Index}
\label{Sec:CDF}
Here, we present the closed derivation formula for the hyperbolic Sombor index using the operators listed in Table~\ref{Table:Ope}. We prove the derivation formula by using the M-polynomial of a graph, which is given below.
\begin{theorem}
	\label{Th:1}
	Let $G=(V(G),E(G))$ be a graph and its hyperbolic Sombor index is
	$$\textit{HSO}(G)=\sum_{uv \in E(G)}f(d(u),d(v)),~\text{where}~ f(x,y)=\frac{\sqrt{x^2+y^2}}{x},~\text{where}~0 <x \le y$$ then
	$$\textit{HSO}(G)= D^{1/2}_{x}JP_{y}P_{x}S_{x}(M(G;x,y))|_{x=1},$$
	where $M(G(x,y))$ is the M-polynomial of $G$.
\end{theorem}
\begin{proof}
	Let $M(G;x,y)$ be the M-polynomial of $G$. Then, using the operators mentioned in Table~\ref{Table:Ope}, we get
\begin{align*}
		D^{1/2}_{x}JP_{y}P_{x}S_{x}(M(G;x,y))&= D^{1/2}_{x}JP_{y}P_{x}S_{x} \biggl\{\sum_{\delta \le i \le j \le \Delta} m_{ij}x^iy^j\biggr\}\\
		&= \sum_{\delta \le i \le j \le \Delta} D^{1/2}_{x}JP_{y}P_{x}S_{x} \{m_{ij}x^iy^j\}\\
		&= \sum_{\delta \le i \le j \le \Delta} D^{1/2}_{x}JP_{y}P_{x} \{\frac{1}{i}m_{ij}x^iy^j\}\\
		&= \sum_{\delta \le i \le j \le \Delta} D^{1/2}_{x}J \{\frac{1}{i}m_{ij}x^{i^2}y^{j^2}\}\\
		&= \sum_{\delta \le i \le j \le \Delta} D^{1/2}_{x} \{\frac{1}{i}m_{ij}x^{i^2+j^2}\}\\
		&=\sum_{\delta \le i \le j \le \Delta}\sqrt{x\cdot \frac{\partial}{\partial x}\{\frac{1}{i}m_{ij}x^{i^2+j^2}\}}\cdot \sqrt{\frac{1}{i}m_{ij}x^{i^2+j^2}}\\
		&=\sum_{\delta \le i \le j \le \Delta}\sqrt{x\cdot\frac{1}{i}m_{ij}\cdot \frac{\partial}{\partial x}\{x^{i^2+j^2}\}}\cdot \sqrt{\frac{1}{i}m_{ij}x^{i^2+j^2}}\\
		&= \sum_{\delta \le i \le j \le \Delta}\frac{1}{i}m_{ij}\sqrt{i^2+j^2}\cdot\sqrt{x\cdot x^{i^2+j^2-1}}\cdot\sqrt{x^{i^2+j^2}}\\
		&= \sum_{\delta \le i \le j \le \Delta} \frac{\sqrt{i^2+j^2}}{i}~m_{ij} x^{i^2+j^2}.
\end{align*}
\begin{align}
	\label{Eq:3}
	\therefore~D^{1/2}_{x}JP_{y}P_{x}S_{x}(M(G;x,y))|_{x=1}&= \sum_{\delta \le i \le j \le \Delta} \frac{\sqrt{i^2+j^2}}{i}~m_{ij} \nonumber\\
	&= \sum_{\delta \le i \le j \le \Delta}m_{ij}\cdot f(i,j).
\end{align}
From the~\Cref{Eq:DBTI,Eq:ADBTI}, we get
\begin{equation}
	\label{Eq:4}
	\textit{HSO}(G)= \sum_{uv \in E(G)}f(d(u),d(v)) = \sum_{\delta \le i \le j \le \Delta}m_{ij}\cdot f(i,j).
\end{equation}
Therefore,~\Cref{Eq:3,Eq:4} give the required result.
\end{proof}
\section{Hyperbolic Sombor Index for Some Well-known Graphs and Chemical Families}
\label{Sec:4}
This section is divided into nine subsections where we calculate the hyperbolic Sombor index for some standard families of graphs, chemical families of $\mathcal{B}_{\alpha}(a,b)$ and dendrimers, jagged-rectangle benzenoid system~$B_{m,n}$, polycyclic aromatic hydrocarbons~$\textit{PAH}_{n}$, V-Phynelenic Nanotubes~$\textit{VPHX}[m,n]$, V-Phynelenic Nanotori~$\textit{VPHY}[m,n]$, porous graphene $\textit{PG}[p,q]$, tadpole graph~$T(n,m)$ and polyphenylenes~$P[s,t]$ using our closed derivation formula of the index over its M-polynomial. Additionally, we illustrate the graphical representation of the M-polynomial of the specified chemical families and their related \textit{HSO} index using MATLAB R2019 software.
\subsection{Hyperbolic Sombor Index for Some Well-known Graphs}
Let us now find the values of the $\textit{HSO}$ index using its derivation formulas of some well-known graphs, such as a complete graph,  complete bipartite graph, star graph, $r$-regular graph, cycle graph and path graph.
\begin{theorem}~\cite{basavanagoud2020}
	\label{Th:KMN}
	Let $G$ be a complete bipartite graph $K_{m,n}$ having $m+n$ vertices $1 \le m \le n$ and $n \ge 2$. Then the M-polynomial of the graph $G$ is $M(G;x,y)= mnx^my^n.$
\end{theorem}
\begin{theorem}
	\label{Th:5}
	If $G$ be a complete bipartite graph $K_{m,n}$ with $1 \le m \le n$ and $n \ge 2$, then
	$$\textit{HSO}(G)= n\cdot\sqrt{m^2+n^2}.$$
\end{theorem}
\begin{proof}
The M-polynomial of the complete bipartite graph $K_{m,n}$ is $M(G;x,y)= mnx^mx^n$ as given in Theorem~\ref{Th:KMN}. Then the hyperbolic Sombor index of $K_{m,n}$ is
\begin{align*}
		\textit{HSO}(G)&= D^{1/2}_{x}JP_{y}P_{x}S_{x}(M(G;x,y))|_{x=1}\\
		&= D^{1/2}_{x}JP_{y}P_{x}S_{x}(mnx^my^n)|_{x=1}\\
		&= n\cdot D^{1/2}_{x}JP_{y}P_{x} (x^my^n)|_{x=1}\\
		&= n\cdot D^{1/2}_{x}J \Big(x^{m^2}y^{n^2}\Big)|_{x=1}\\
		&= n\cdot D^{1/2}_{x} \Big(x^{m^2+n^2}\Big)|_{x=1}\\
		&= n\cdot \sqrt{m^2+n^2} \cdot \Big(x^{m^2+n^2}\Big)|_{x=1}\\
		&= n\cdot\sqrt{m^2+n^2}.
\end{align*}
\end{proof}
\Cref{Cor:7} can be derived from~\Cref{Th:5} by putting $m=n=r$. Similarly, resulting~\Cref{Cor:9} can be achieved by putting $m=1$ and $n=r-1$ in~\Cref{Th:5}.
\begin{corollary}
	\label{Cor:7}
If $G$ be a complete bipartite graph $K_{r,r}$ with $r \ge 2$, then
$\textit{HSO}(G)= \sqrt{2}r^2.$
\end{corollary}
\begin{corollary}
	\label{Cor:9}
	If $G$ be a star graph $K_{1,r-1}$ with $r \ge 2$, then
$$\textit{HSO}(G)= (r-1)\sqrt{r^2-2r+2}.$$
\end{corollary}
\begin{theorem}~\cite{basavanagoud2019}
	\label{Th:RRG}
	Let $G$ be a $r$-regular graph having $n$ vertices and $r \ge 2$. Then the M-polynomial of the $r$-regular graph is given by $M(G;x,y)= \frac{nr}{2}x^ry^r.$
\end{theorem}
\begin{theorem}
	\label{Th:11}
	If $G$ be a $r$-regular graph having $n$ vertices and $r \ge 2$, then
	$$\textit{HSO}(G)=\frac{nr}{\sqrt{2}}.$$
\end{theorem}
\begin{proof}
The M-polynomial of the $r$-regular graph is $M(G;x,y)= \frac{nr}{2}x^ry^r$ as given in~Theorem~\ref{Th:RRG}. Then the hyperbolic Sombor index of $r$-regular graph is
\begin{align*}
		\textit{HSO}(G)&= D^{1/2}_{x}JP_{y}P_{x}S_{x}(M(G;x,y))|_{x=1}\\
		&= D^{1/2}_{x}JP_{y}P_{x}S_{x}\bigg(\frac{nr}{2}x^ry^r\bigg)|_{x=1}\\
		&= \bigg(\frac{n}{2}\bigg)\cdot D^{1/2}_{x}JP_{y}P_{x}(x^ry^r)|_{x=1}\\
		&= \bigg(\frac{n}{2}\bigg) \cdot D^{1/2}_{x}J(x^{r^2}y^{r^2})|_{x=1}\\
		&= \bigg(\frac{n}{2}\bigg) \cdot D^{1/2}_{x}(x^{2r^2})|_{x=1}\\
		&= \bigg(\frac{n}{2}\bigg) \sqrt{2r^2} \cdot (x^{2r^2})|_{x=1}\\
		&= \bigg(\frac{n}{2}\bigg) \sqrt{2r^2}\\
		&= \frac{nr}{\sqrt{2}}.
\end{align*}
\end{proof}
The proof of the following~\Cref{Cor:13,Cor:15} are quite evident and can be solved similarly to~\Cref{Th:5,Th:11}. They are left as an exercise for the readers.
\begin{theorem}~\cite{basavanagoud2020}
	Let $G$ be a cycle graph $C_n$ with $n (\ge 3)$ vertices, then the M-polynomial of $G$ is $M(G;x,y)= nx^2y^2.$
\end{theorem}
\begin{corollary}
	\label{Cor:13}
	If $G$ be a cycle graph $C_n$ with $n (\ge 3)$ vertices, then $\textit{HSO}(G)= \sqrt{2}n.$
\end{corollary}
\begin{theorem}~\cite{basavanagoud2020}
	Consider a complete graph $K_n$, a $(n-1)$-regular graph with $n(\ge 3)$ vertices. Then the M-polynomial of the graph $G$ is $M(G;x,y)= \frac{n(n-1)}{2}x^{n-1}y^{n-1}.$
\end{theorem}
\begin{corollary}
	\label{Cor:15}
If $G$ be a complete graph $K_n$ with $n (\ge 3)$ vertices, then
$$\textit{HSO}(G)= \frac{n(n-1)}{\sqrt{2}}.$$
\end{corollary}
\begin{theorem}~\cite{basavanagoud2020}
	\label{Th:PG}
	Let $G$ be a path graph $P_n$ with $n (\ge 3)$ vertices, then the M-polynomial of $G$ is $M(G;x,y)= 2x^1y^2+(n-3)x^2y^2.$
\end{theorem}
\begin{theorem}
If $G$ be a path graph $P_n$ with $n (\ge 3)$ vertices, then
$\textit{HSO}(G)= 2\sqrt{5}+\sqrt{2}(n-3).$
\end{theorem}

\begin{proof}
The M-polynomial of the path graph is $M(G;x,y)= 2x^1y^2+(n-3)x^2y^2$ as given in~Theorem~\ref{Th:PG}. Then the hyperbolic Sombor index of $P_n$ is
\begin{align*}
		\textit{HSO}(G)&= D^{1/2}_{x}JP_{y}P_{x}S_{x}(M(G;x,y))|_{x=1}\\
		&= D^{1/2}_{x}JP_{y}P_{x}S_{x}\Big(2x^1y^2+(n-3)x^2y^2\Big)|_{x=1}\\
		&= D^{1/2}_{x}JP_{y}P_{x}\Big(2x^1y^2+\Big(\frac{n-3}{2}\Big)x^2y^2\Big)|_{x=1}\\
		&= D^{1/2}_{x}J\Big(2x^1y^4+\Big(\frac{n-3}{2}\Big)x^4y^4\Big)|_{x=1}\\
		&= D^{1/2}_{x}\Big(2x^5+\Big(\frac{n-3}{2}\Big)x^8\Big)|_{x=1}\\
		&= \Big(2\sqrt{5}x^5+\sqrt{2}(n-3)x^8\Big)|_{x=1}\\
		&= 2\sqrt{5}+\sqrt{2}(n-3).
\end{align*}
\end{proof}
\subsection{Hyperbolic Sombor Index for Boron Icosahedral $\boldsymbol{\alpha}$ sheet $\boldsymbol{\mathcal{B}_{\alpha}(a,b)}$}
This subsection calculates the hyperbolic Sombor index for the boron icosahedral $\alpha$ sheet $\mathcal{B}_{\alpha}(a,b)$. Note that, $\mathcal{B}_{\alpha}(a,b)$ contains $96ab$ and $268ab-9a-9b+2$ number of vertices and edges, respectively, while the degree of each vertex is either $5$ or $6$. Using the M-polynomial provided in Theorem~\ref{Th:14}, we compute the hyperbolic Sombor index for $\mathcal{B}_{\alpha}(a,b)$.
\begin{theorem}~\cite{xavier2024computing}
	\label{Th:14}
Let $G$ be the family of boron icosahedral $\alpha$ sheet $\mathcal{B}_{\alpha}(a,b)$. Then,
\begin{equation*}
		M(G;x,y)=(46ab+46a-48b-4)x^5y^5+(108ab-2a-6b-12)x^5y^6+(114ab-53a-51b+18)x^6y^6.
\end{equation*}
\end{theorem}
\begin{theorem}
Let $G$ be the boron icosahedral $\alpha$ sheet $\mathcal{B}_{\alpha}(a,b)$. Then the hyperbolic Sombor index is given by
\begin{align*}
		\textit{HSO}(G)&=\sqrt{2}(46ab+46a-48b-4)+\frac{\sqrt{61}}{5}(108ab-2a-6b-12)+\sqrt{2}(114ab-53a\\
		&\ \ \ -51b+18).
\end{align*}
\end{theorem}

\begin{proof}
	The M-polynomial of $\mathcal{B}_{\alpha}(a,b)$ is given by
	\begin{equation*}
		M(G;x,y)=(46ab+46a-48b-4)x^5y^5+(108ab-2a-6b-12)x^5y^6+(114ab-53a-51b+18)x^6y^6.
	\end{equation*}
Then 
\begin{align*}
    &\textit{HSO}(G)\\
	&= D^{1/2}_{x}JP_{y}P_{x}S_{x}(M(G;x,y))|_{x=1}\\
	&=D^{1/2}_{x}JP_{y}P_{x}S_{x}\Big((46ab+46a-48b-4)x^5y^5+(108ab-2a-6b-12)x^5y^6+(114ab-53a\\
	&\ \ \ -51b+18)x^6y^6\Big)|_{x=1}\\
	&=D^{1/2}_{x}JP_{y}P_{x}\Bigg(\Big(\frac{46ab+46a-48b-4}{5}\Big)x^5y^5+\Big(\frac{108ab-2a-6b-12}{5}\Big)x^5y^6\\
	&\ \ \ +\Big(\frac{114ab-53a-51b+18}{6}\Big)x^6y^6\Bigg)|_{x=1}\\
	&=D^{1/2}_{x}J\Bigg(\Big(\frac{46ab+46a-48b-4}{5}\Big)x^{25}y^{25}+\Big(\frac{108ab-2a-6b-12}{5}\Big)x^{25}y^{36}\\
	&\ \ \ +\Big(\frac{114ab-53a-51b+18}{6}\Big)x^{36}y^{36}\Bigg)|_{x=1}\\
	&=D^{1/2}_{x}\Bigg(\Big(\frac{46ab+46a-48b-4}{5}\Big)x^{50}+\Big(\frac{108ab-2a-6b-12}{5}\Big)x^{61}\\
	&\ \ \ +\Big(\frac{114ab-53a-51b+18}{6}\Big)x^{72}\Bigg)|_{x=1}\\
	&=\Bigg(\sqrt{50}\cdot\Big(\frac{46ab+46a-48b-4}{5}\Big)x^{50}+\sqrt{61}\cdot\Big(\frac{108ab-2a-6b-12}{5}\Big)x^{61}\\
	&\ \ \ +\sqrt{72}\cdot\Big(\frac{114ab-53a-51b+18}{6}\Big)x^{72}\Bigg)|_{x=1}\\
	&=\Bigg(\sqrt{50}\cdot\Big(\frac{46ab+46a-48b-4}{5}\Big)+\sqrt{61}\cdot\Big(\frac{108ab-2a-6b-12}{5}\Big)\\
	&\ \ \ +\sqrt{72}\cdot\Big(\frac{114ab-53a-51b+18}{6}\Big)\Bigg)\\
	&= \sqrt{2}(46ab+46a-48b-4)+\frac{\sqrt{61}}{5}(108ab-2a-6b-12)+\sqrt{2}(114ab-53a-51b+18).
\end{align*}
\end{proof}
\begin{figure}[htb!]
	\centering
	\begin{minipage}[t]{0.45\textwidth}
		\includegraphics[width=\textwidth]{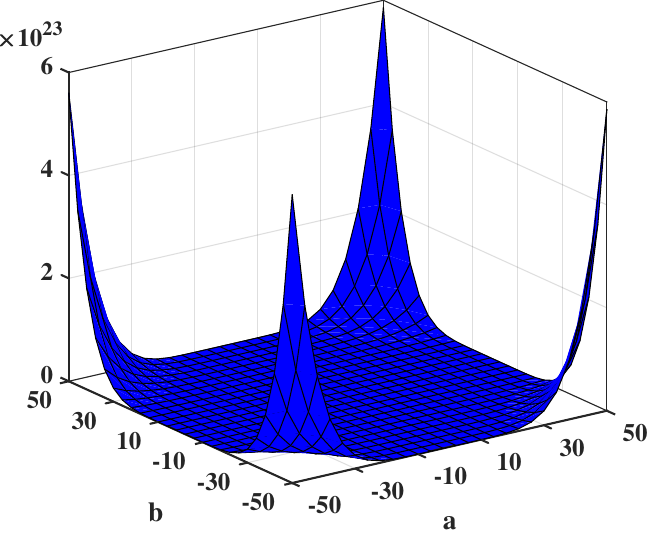}
		\centerline{\textbf{(a)}}
		\label{Fig:MP_Bab}
	\end{minipage}
	\hspace{0.2cm}
	\begin{minipage}[t]{0.45\textwidth}
		\includegraphics[width=\textwidth]{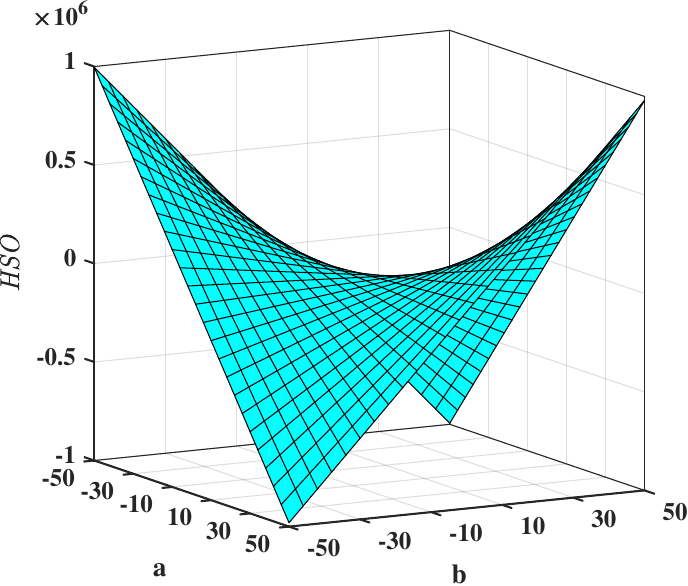}
		\centerline{\textbf{(b)}}
		\label{Fig:HSO_Bab}
	\end{minipage}
	\caption{Graphical illustration of the (a)~M-polynomial of $\mathcal{B}_{\alpha}(5,5)$ and (b)~\textit{HSO} index of $\mathcal{B}_{\alpha}(a,b)$.}
\end{figure}
\subsection{Hyperbolic Sombor Index for Dendrimers}
In $1985$, Tomalia et al. made the initial discovery of dendrimers~\cite{tomalia1985new}, a novel type of highly organized and branching polymeric material. They usually have a three-dimensional spherical shape and are symmetrical around the core. The unique architecture of dendrimers profoundly influences their chemical and physical properties, enabling their broad application across nanoscience, biomedical fields and industrial chemistry. In this context, we examine four prominent dendrimer families: poly(propyl) ether imine (\textit{PETIM}), porphyrin ($D_nP_n$), zinc porphyrin ($\textit{DPZ}_n$) and poly ethylene amide amine (\textit{PETAA}) dendrimers.
Then, using the M-polynomial provided in~\Cref{Th:16,Th:18,Th:20,Th:22}, we calculate the hyperbolic Sombor index for each dendrimer.
\begin{theorem}~\cite{sdas2023a}
	\label{Th:16}
	Let $G$ be the family of poly(propyl) ether imine (PETIM) dendrimers, then the expression of M-polynomial is
	\begin{equation*}
		M(G;x,y)= 2^{n+1}xy^2+(2^{n+4}-18)x^2y^2+(3\times 2^{n+1}-6)x^2y^3.
	\end{equation*}
\end{theorem}
\begin{theorem}
	Let $G$ be the family of poly(propyl) ether imine (PETIM) dendrimers. Then the hyperbolic Sombor index is given by
	$$\textit{HSO}(G)=\sqrt{5}\times 2^{n}+\sqrt{8}(2^{n+3}-9)+\sqrt{13}(3\times 2^n-3).$$
\end{theorem}
\begin{proof}
	The M-polynomial of the \textit{PETIM} dendrimer is given by
	\begin{equation*}
		M(G;x,y)=2^{n+1}xy^2+(2^{n+4}-18)x^2y^2+(3\times 2^{n+1}-6)x^2y^3.
	\end{equation*}
Then
\begin{align*}
	\textit{HSO}(G)&= D^{1/2}_{x}JP_{y}P_{x}S_{x}(M(G;x,y))|_{x=1}\\
	&= D^{1/2}_{x}JP_{y}P_{x}S_{x}\Big(2^{n+1}xy^2+(2^{n+4}-18)x^2y^2+(3\times 2^{n+1}-6)x^2y^3\Big)|_{x=1}\\
	&= D^{1/2}_{x}JP_{y}P_{x}\Big(2^{n}\times xy^2+(2^{n+3}-9)x^2y^2+(3\times 2^{n}-3)x^2y^3\Big)|_{x=1}\\
	&= D^{1/2}_{x}J\Big(2^{n}\times xy^4+(2^{n+3}-9)x^4y^4+(3\times 2^{n}-3)x^4y^9\Big)|_{x=1}\\
	&= D^{1/2}_{x}\Big(2^{n}\times x^5+(2^{n+3}-9)x^8+(3\times 2^{n}-3)x^{13}\Big)|_{x=1}\\
	&= \Big(\sqrt{5}\times 2^{n}\times x^5+\sqrt{8}(2^{n+3}-9)x^8+\sqrt{13}(3\times 2^{n}-3)x^{13}\Big)|_{x=1}\\
	&=\sqrt{5}\times 2^{n}+\sqrt{8}(2^{n+3}-9)+\sqrt{13}(3\times 2^n-3).
\end{align*}
\end{proof}
\begin{figure}[htb!]
	\centering
	\begin{minipage}[t]{0.45\textwidth}
		\includegraphics[width=\textwidth]{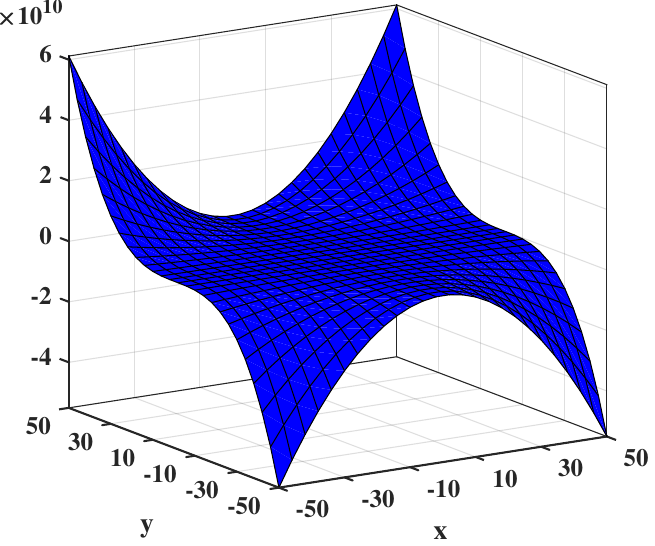}
		\centerline{\textbf{(a)}}
		\label{Fig:MP_PETIM}
	\end{minipage}
	\hspace{0.2cm}
	\begin{minipage}[t]{0.45\textwidth}
		\includegraphics[width=\textwidth]{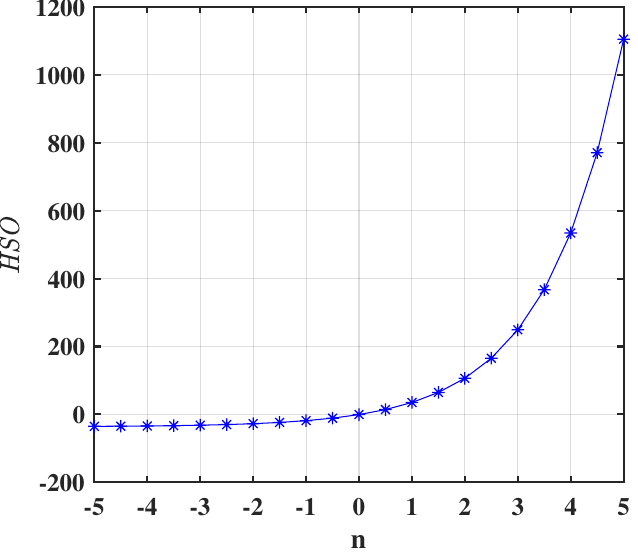}
		\centerline{\textbf{(b)}}
		\label{Fig:HSO_PETIM}
	\end{minipage}
	\caption{Graphical illustration of the (a)~M-polynomial and (b)~\textit{HSO} index of \textit{PETIM} dendrimer for $n=5$.}
\end{figure}
\begin{theorem}~\cite{sdas2023a}
	\label{Th:18}
	Let $G$ be the family of porphyrin ($D_nP_n$) dendrimers, then the expression of M-polynomial is
	\begin{equation*}
		M(G;x,y)= 2nxy^3+24nxy^4+(10n-5)x^2y^2+(48n-6)x^2y^3+13nx^3y^3+8nx^3y^4.
	\end{equation*}
\end{theorem}
\begin{theorem}
	Let $G$ be the family of porphyrin ($D_nP_n$) dendrimers. Then the hyperbolic Sombor index is given by
	$$\textit{HSO}(G)=2\sqrt{10}n+24\sqrt{17}n+\sqrt{2}(10n-5)+\sqrt{13}(24n-3)+13\sqrt{2}n+\Big(\frac{40n}{3}\Big).$$
\end{theorem}
\begin{proof}
	The M-polynomial of the $D_nP_n$ dendrimer is given by
	\begin{equation*}
		M(G;x,y)=2nxy^3+24nxy^4+(10n-5)x^2y^2+(48n-6)x^2y^3+13nx^3y^3+8nx^3y^4.
	\end{equation*}
	Then
	\begin{align*}
		\textit{HSO}(G)&= D^{1/2}_{x}JP_{y}P_{x}S_{x}(M(G;x,y))|_{x=1}\\
		&=D^{1/2}_{x}JP_{y}P_{x}S_{x}\Big(2nxy^3+24nxy^4+(10n-5)x^2y^2+(48n-6)x^2y^3+13nx^3y^3\\
		&\ \ \ +8nx^3y^4\Big)|_{x=1}\\
		&=D^{1/2}_{x}JP_{y}P_{x}\Big(2nxy^3+24nxy^4+\Big(\frac{10n-5}{2}\Big)x^2y^2+(24n-3)x^2y^3+\Big(\frac{13n}{3}\Big)x^3y^3\\
		&\ \ \ +\Big(\frac{8n}{3}\Big)x^3y^4\Big)|_{x=1}\\
		&=D^{1/2}_{x}J\Big(2nxy^9+24nxy^{16}+\Big(\frac{10n-5}{2}\Big)x^4y^4+(24n-3)x^4y^9+\Big(\frac{13n}{3}\Big)x^9y^9\\
		&\ \ \ +\Big(\frac{8n}{3}\Big)x^9y^{16}\Big)|_{x=1}\\
		&=D^{1/2}_{x}\Big(2nx^{10}+24nx^{17}+\Big(\frac{10n-5}{2}\Big)x^8+(24n-3)x^{13}+\Big(\frac{13n}{3}\Big)x^{18}\\
		&\ \ \ +\Big(\frac{8n}{3}\Big)x^{25}\Big)|_{x=1}\\
		&=\Big(2\sqrt{10}nx^{10}+24\sqrt{17}nx^{17}+\sqrt{2}(10n-5)x^8+\sqrt{13}(24n-3)x^{13}+13\sqrt{2}nx^{18}\\
		&\ \ \ +\Big(\frac{40n}{3}\Big)x^{25}\Big)|_{x=1}\\
		&=2\sqrt{10}n+24\sqrt{17}n+\sqrt{2}(10n-5)+\sqrt{13}(24n-3)+13\sqrt{2}n+\Big(\frac{40n}{3}\Big).
	\end{align*}
\end{proof}
\begin{figure}[htb!]
	\centering
	\begin{minipage}[t]{0.45\textwidth}
		\includegraphics[width=\textwidth]{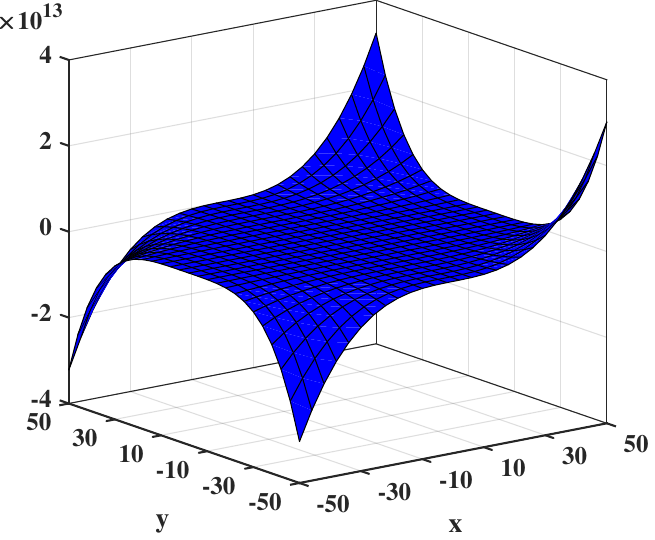}
		\centerline{\textbf{(a)}}
		\label{Fig:MP_DP}
	\end{minipage}
	\hspace{0.2cm}
	\begin{minipage}[t]{0.45\textwidth}
		\includegraphics[width=\textwidth]{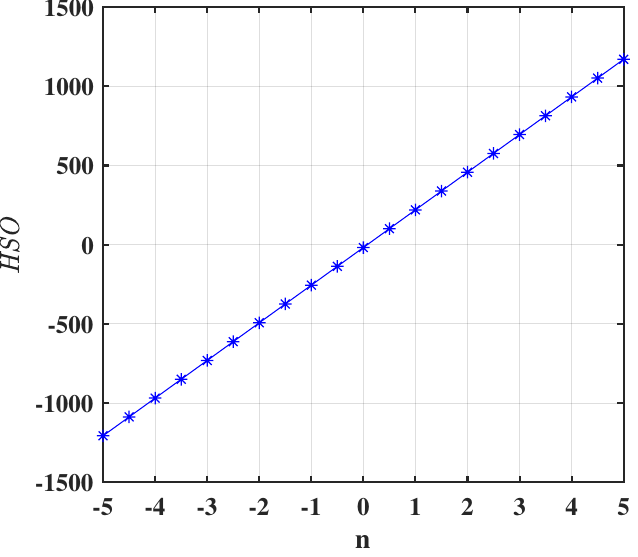}
		\centerline{\textbf{(b)}}
		\label{Fig:HSO_DP}
	\end{minipage}
	\caption{Graphical illustration of the (a)~M-polynomial and (b)~\textit{HSO} index of $D_nP_n$ dendrimer for $n=5$.}
\end{figure}
\begin{theorem}~\cite{sdas2023a}
	\label{Th:20}
	Let $G$ be the family of zinc porphyrin ($\textit{DPZ}_n$) dendrimers, then the expression of M-polynomial is
	\begin{equation*}
		M(G;x,y)=(16\times 2^n-4)x^2y^2+(40\times 2^n-16)x^2y^3+(8\times 2^n +12)x^3y^3+4x^3y^4.
	\end{equation*}
\end{theorem}
\begin{theorem}
	Let $G$ be the family of zinc porphyrin ($\textit{DPZ}_n$) dendrimers. Then the hyperbolic Sombor index is given by
	$$\textit{HSO}(G)=\sqrt{8}(2^{n+3}-2)+\sqrt{13}(20\times 2^n-8)+\sqrt{18}\Big(\frac{2^{n+3}}{3} +4\Big)+\Big(\frac{20}{3}\Big).$$
\end{theorem}
\begin{proof}
	The M-polynomial of the $\textit{DPZ}_n$ dendrimer is given by
	\begin{equation*}
		M(G;x,y)=(16\times 2^n-4)x^2y^2+(40\times 2^n-16)x^2y^3+(8\times 2^n +12)x^3y^3+4x^3y^4.
	\end{equation*}
	Then
	\begin{align*}
		&\textit{HSO}(G)\\
		&= D^{1/2}_{x}JP_{y}P_{x}S_{x}(M(G;x,y))|_{x=1}\\
		&=D^{1/2}_{x}JP_{y}P_{x}S_{x}\Big((16\times 2^n-4)x^2y^2+(40\times 2^n-16)x^2y^3+(8\times 2^n +12)x^3y^3+4x^3y^4\Big)|_{x=1}\\
		&=D^{1/2}_{x}JP_{y}P_{x}\Big((8\times 2^n-2)x^2y^2+(20\times 2^n-8)x^2y^3+\Big(\frac{8}{3}\times 2^n +4\Big)x^3y^3+\Big(\frac{4}{3}\Big)x^3y^4\Big)|_{x=1}\\
		&=D^{1/2}_{x}J\Big((8\times 2^n-2)x^4y^4+(20\times 2^n-8)x^4y^9+\Big(\frac{8}{3}\times 2^n +4\Big)x^9y^9+\Big(\frac{4}{3}\Big)x^9y^{16}\Big)|_{x=1}\\
		&=D^{1/2}_{x}\Big((8\times 2^n-2)x^8+(20\times 2^n-8)x^{13}+\Big(\frac{8}{3}\times 2^n +4\Big)x^{18}+\Big(\frac{4}{3}\Big)x^{25}\Big)|_{x=1}\\
		&=\Big(\sqrt{8}(8\times 2^n-2)x^8+\sqrt{13}(20\times 2^n-8)x^{13}+\sqrt{18}\Big(\frac{8}{3}\times 2^n +4\Big)x^{18}+\Big(\frac{20}{3}\Big)x^{25}\Big)|_{x=1}\\
		&=\sqrt{8}(2^{n+3}-2)+\sqrt{13}(20\times 2^n-8)+\sqrt{18}\Big(\frac{2^{n+3}}{3} +4\Big)+\Big(\frac{20}{3}\Big).
	\end{align*}
\end{proof}
\begin{figure}[htb!]
	\centering
	\begin{minipage}[t]{0.45\textwidth}
		\includegraphics[width=\textwidth]{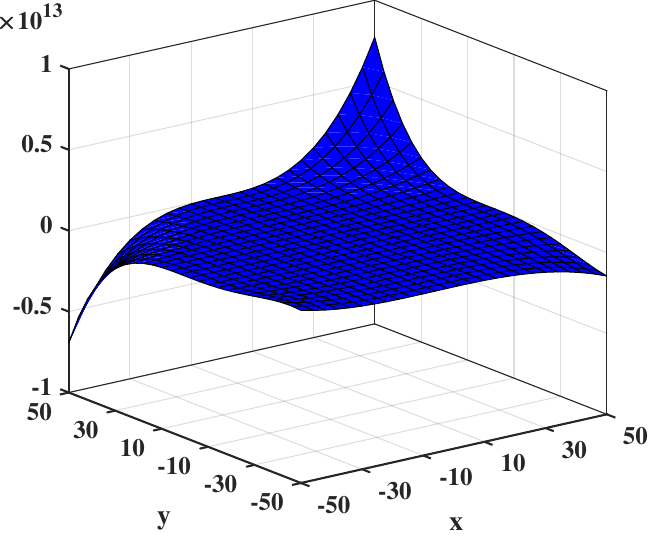}
		\centerline{\textbf{(a)}}
		\label{Fig:MP_DPZ}
	\end{minipage}
	\hspace{0.2cm}
	\begin{minipage}[t]{0.45\textwidth}
		\includegraphics[width=\textwidth]{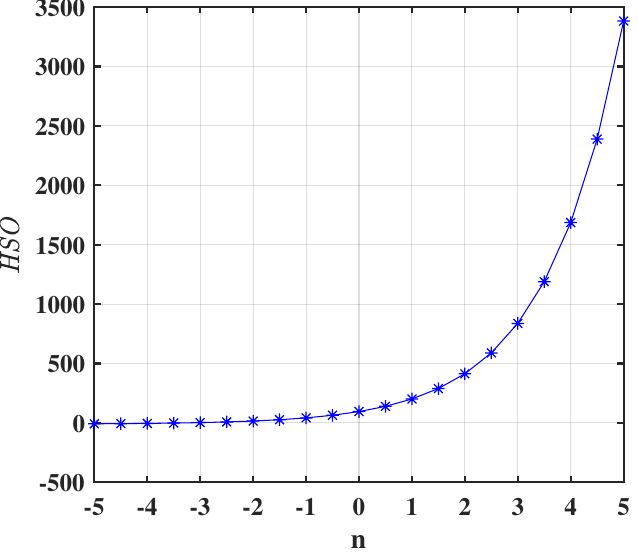}
		\centerline{\textbf{(b)}}
		\label{Fig:HSO_DPZ}
	\end{minipage}
	\caption{Graphical illustration of the (a)~M-polynomial and (b)~\textit{HSO} index of $\textit{DPZ}_n$ dendrimer for $n=5$.}
\end{figure}
\begin{theorem}~\cite{sdas2023a}
	\label{Th:22}
	Let $G$ be the family of poly ethylene amide amine (PETAA) dendrimers, then the expression of M-polynomial is
	\begin{equation*}
		M(G;x,y)=2^{n+2}xy^2+(2^{n+2}-2)xy^3+(2^{n+4}-8)x^2y^2+(20\times 2^n-9)x^2y^3.
	\end{equation*}
\end{theorem}
\begin{theorem}
	Let $G$ be the family of poly ethylene amide amine (PETAA) dendrimers. Then the hyperbolic Sombor index is given by
	$$\textit{HSO}(G)=\sqrt{5}\cdot 2^{n+2}+\sqrt{10}(2^{n+2}-2)+\sqrt{8}(2^{n+3}-4)+\sqrt{13}\Big(10\times 2^n-\frac{9}{2}\Big).$$
\end{theorem}
\begin{proof}
	The M-polynomial of the \textit{PETAA} dendrimer is given by
	\begin{equation*}
		M(G;x,y)=2^{n+2}xy^2+(2^{n+2}-2)xy^3+(2^{n+4}-8)x^2y^2+(20\times 2^n-9)x^2y^3.
	\end{equation*}
	Then
	\begin{align*}
		&\textit{HSO}(G)\\
		&= D^{1/2}_{x}JP_{y}P_{x}S_{x}(M(G;x,y))|_{x=1}\\
		&=D^{1/2}_{x}JP_{y}P_{x}S_{x}\Big(2^{n+2}xy^2+(2^{n+2}-2)xy^3+(2^{n+4}-8)x^2y^2+(20\times 2^n-9)x^2y^3\Big)|_{x=1}\\
		&=D^{1/2}_{x}JP_{y}P_{x}\Big(2^{n+2}xy^2+(2^{n+2}-2)xy^3+(2^{n+3}-4)x^2y^2+\Big(10\times 2^n-\frac{9}{2}\Big)x^2y^3\Big)|_{x=1}\\
		&=D^{1/2}_{x}J\Big(2^{n+2}xy^4+(2^{n+2}-2)xy^9+(2^{n+3}-4)x^4y^4+\Big(10\times 2^n-\frac{9}{2}\Big)x^4y^9\Big)|_{x=1}\\
		&=D^{1/2}_{x}\Big(2^{n+2}x^5+(2^{n+2}-2)x^{10}+(2^{n+3}-4)x^8+\Big(10\times 2^n-\frac{9}{2}\Big)x^{13}\Big)|_{x=1}\\
		&=\Big(\sqrt{5}\cdot 2^{n+2}x^5+\sqrt{10}(2^{n+2}-2)x^{10}+\sqrt{8}(2^{n+3}-4)x^8+\sqrt{13}\Big(10\times 2^n-\frac{9}{2}\Big)x^{13}\Big)|_{x=1}\\
		&=\sqrt{5}\cdot 2^{n+2}+\sqrt{10}(2^{n+2}-2)+\sqrt{8}(2^{n+3}-4)+\sqrt{13}\Big(10\times 2^n-\frac{9}{2}\Big).
	\end{align*}
\end{proof}
\begin{figure}[htb!]
	\centering
	\begin{minipage}[t]{0.45\textwidth}
		\includegraphics[width=\textwidth]{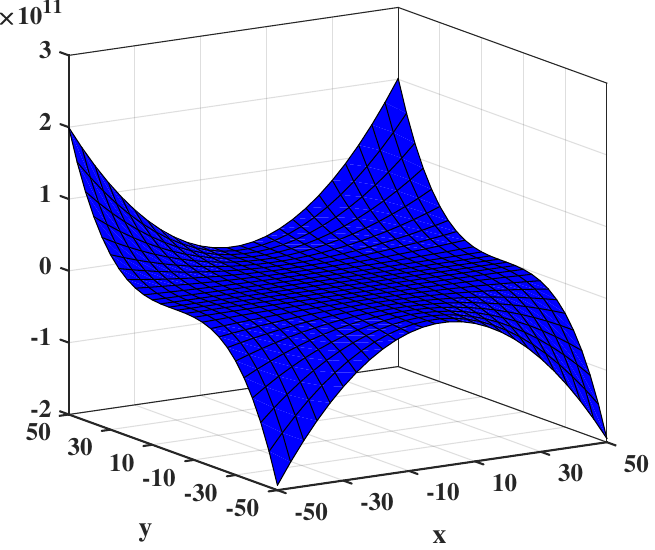}
		\centerline{\textbf{(a)}}
		\label{Fig:MP_PETAA}
	\end{minipage}
	\hspace{0.2cm}
	\begin{minipage}[t]{0.45\textwidth}
		\includegraphics[width=\textwidth]{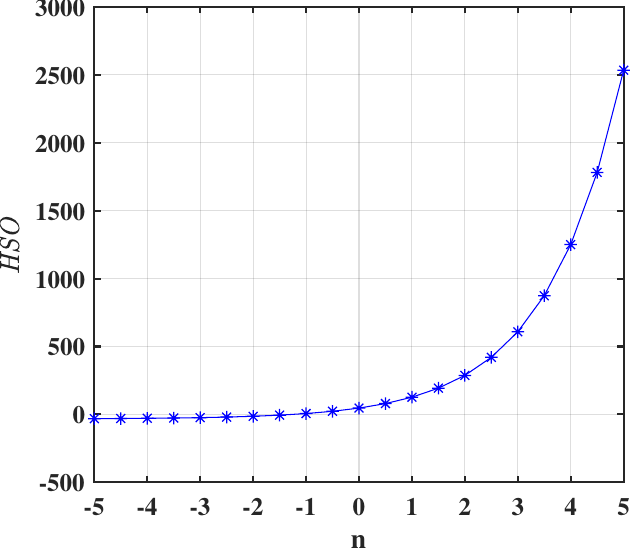}
		\centerline{\textbf{(b)}}
		\label{Fig:HSO_PETAA}
	\end{minipage}
	\caption{Graphical illustration of the (a)~M-polynomial and (b)~\textit{HSO} index of \textit{PETAA} dendrimer for $n=5$.}
\end{figure}
\subsection{Hyperbolic Sombor Index for Jagged-rectangle Benzenoid System $\boldsymbol{B_{m,n}}$}
Shui Ling-Ling et al.~\cite{shui2005} defined the jagged-rectangle benzenoid system $B_{m,n}$ in $2005$. The number of hexagonal cells in each chain alternates between $m$ and $m-1$; thus, the shape of the hexagonal jagged rectangle $B_{m,n}$ is a rectangle. There are $4mn+4m+2n-2$ vertices and $6mn+5m+n-4$ edges in $B_{m,n}$, respectively, and each vertex has a degree of $2$ or $3$.
\begin{theorem}~\cite{kwun2018}
	Let $G$ be the family of jagged-rectangle benzenoid system $B_{m,n}$ with $m \in \mathbb{N}\setminus \{1\}$ and $n \in \mathbb{N}$. Then the expression of M-polynomial is
	\begin{equation*}
		M(G;x,y)=(2n+4)x^2y^2+(4m+4n-4)x^2y^3+(6mn+m-5n-4)x^3y^3.
	\end{equation*}
\end{theorem}
\begin{theorem}
	Let $G$ be the family of jagged-rectangle benzenoid system ($B_{m,n}$). Then the hyperbolic Sombor index is given by
	$$\textit{HSO}(G)=\sqrt{8}(n+2)+\sqrt{13}(2m+2n-2)+\sqrt{18}\Big(2mn+\frac{m}{3}-\frac{5n}{3}-\frac{4}{3}\Big).$$
\end{theorem}
\begin{proof}
	The M-polynomial of $B_{m,n}$ is given by
	\begin{equation*}
		M(G;x,y)=(2n+4)x^2y^2+(4m+4n-4)x^2y^3+(6mn+m-5n-4)x^3y^3.
	\end{equation*}
	Then
	\begin{align*}
		&\textit{HSO}(G)\\
		&= D^{1/2}_{x}JP_{y}P_{x}S_{x}(M(G;x,y))|_{x=1}\\
		&=D^{1/2}_{x}JP_{y}P_{x}S_{x}\Big((2n+4)x^2y^2+(4m+4n-4)x^2y^3+(6mn+m-5n-4)x^3y^3\Big)|_{x=1}\\
		&=D^{1/2}_{x}JP_{y}P_{x}\Big((n+2)x^2y^2+(2m+2n-2)x^2y^3+\Big(2mn+\frac{m}{3}-\frac{5n}{3}-\frac{4}{3}\Big)x^3y^3\Big)|_{x=1}\\
		&=D^{1/2}_{x}J\Big((n+2)x^4y^4+(2m+2n-2)x^4y^9+\Big(2mn+\frac{m}{3}-\frac{5n}{3}-\frac{4}{3}\Big)x^9y^9\Big)|_{x=1}\\
		&=D^{1/2}_{x}\Big((n+2)x^8+(2m+2n-2)x^{13}+\Big(2mn+\frac{m}{3}-\frac{5n}{3}-\frac{4}{3}\Big)x^{18}\Big)|_{x=1}\\
		&=\Big(\sqrt{8}(n+2)x^8+\sqrt{13}(2m+2n-2)x^{13}+\sqrt{18}\Big(2mn+\frac{m}{3}-\frac{5n}{3}-\frac{4}{3}\Big)x^{18}\Big)|_{x=1}\\
		&=\sqrt{8}(n+2)+\sqrt{13}(2m+2n-2)+\sqrt{18}\Big(2mn+\frac{m}{3}-\frac{5n}{3}-\frac{4}{3}\Big).
	\end{align*}
\end{proof}
\begin{figure}[htb!]
	\centering
	\begin{minipage}[t]{0.45\textwidth}
		\includegraphics[width=\textwidth]{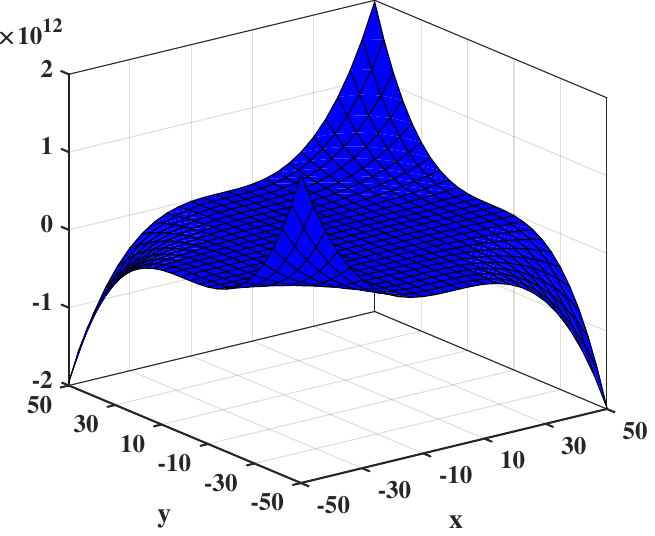}
		\centerline{\textbf{(a)}}
		\label{Fig:MP_Bmn}
	\end{minipage}
	\hspace{0.2cm}
	\begin{minipage}[t]{0.45\textwidth}
		\includegraphics[width=\textwidth]{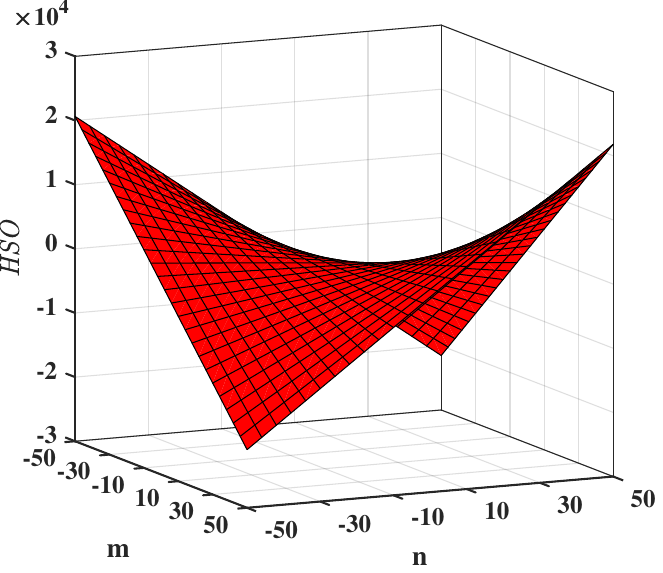}
		\centerline{\textbf{(b)}}
		\label{Fig:HSO_Bmn}
	\end{minipage}
	\caption{Graphical illustration of the (a)~M-polynomial of $B_{m,n}$ for $m=n=5$ and (b)~\textit{HSO} index of $B_{m,n}$.}
\end{figure}
\subsection{Hyperbolic Sombor Index for Polycyclic Aromatic Hydrocarbon $\textit{PAH}_n$}
The polycyclic aromatic hydrocarbon ($\textit{PAH}_n$) is a highly desirable structure in several scientific domains, particularly physics, chemistry and nanoscience. For additional findings and details about the structure, reader may follow the references~\cite{shao2018computing,jamil2016computing,dietz20002}. The molecular graph of $\textit{PAH}_n$ consists of $6n^2+6n$ vertices and $9n^2+3n$ edges.
\begin{theorem}
	Let $G$ be the family of polycyclic aromatic hydrocarbons ($\textit{PAH}_n$), then the expression of M-polynomial is
	\begin{equation*}
		M(G;x,y)=6nx^1y^3+(9n^2-3n)x^3y^3.
	\end{equation*}
\end{theorem}
\begin{proof}
	The graph $G$ contains two types of edges in terms of the degree of end vertices as follows:
	\begin{align*}
		& E_{1,3}=\{uv \in E(G): d(u)=1,d(v)=3\}, & |E_{1,3}|=6n,\\
		& E_{3,3}=\{uv \in E(G): d(u)=d(v)=3\}, & |E_{3,3}|=9n^2-3n.
	\end{align*}
	Then by applying the Definition~\ref{Def:MP}, we get the M-polynomial of polycyclic aromatic hydrocarbons $\textit{PAH}_n$.
\end{proof}
\begin{theorem}
	Let $G$ be the family of polycyclic aromatic hydrocarbons ($\textit{PAH}_n$). Then the hyperbolic Sombor index is given by
	$$\textit{HSO}(G)=6\sqrt{10}n+\sqrt{18}(3n^2-n).$$
\end{theorem}
\begin{proof}
	The M-polynomial of $\textit{PAH}_n$ is given by
	\begin{equation*}
		M(G;x,y)=6nx^1y^3+(9n^2-3n)x^3y^3.
	\end{equation*}
	Then
	\begin{align*}
		\textit{HSO}(G)&= D^{1/2}_{x}JP_{y}P_{x}S_{x}(M(G;x,y))|_{x=1}\\
		&=D^{1/2}_{x}JP_{y}P_{x}S_{x}\Big(6nx^1y^3+(9n^2-3n)x^3y^3\Big)|_{x=1}\\
		&=D^{1/2}_{x}JP_{y}P_{x}\Big(6nx^1y^3+(3n^2-n)x^3y^3\Big)|_{x=1}\\
		&=D^{1/2}_{x}J\Big(6nx^1y^9+(3n^2-n)x^9y^9\Big)|_{x=1}\\
		&=D^{1/2}_{x}\Big(6nx^{10}+(3n^2-n)x^{18}\Big)|_{x=1}\\
		&=\Big(\sqrt{10}\cdot6nx^{10}+\sqrt{18}(3n^2-n)x^{18}\Big)|_{x=1}\\
		&=6\sqrt{10}n+\sqrt{18}(3n^2-n).
	\end{align*}
\end{proof}
\begin{figure}[htb!]
	\centering
	\begin{minipage}[t]{0.45\textwidth}
		\includegraphics[width=\textwidth]{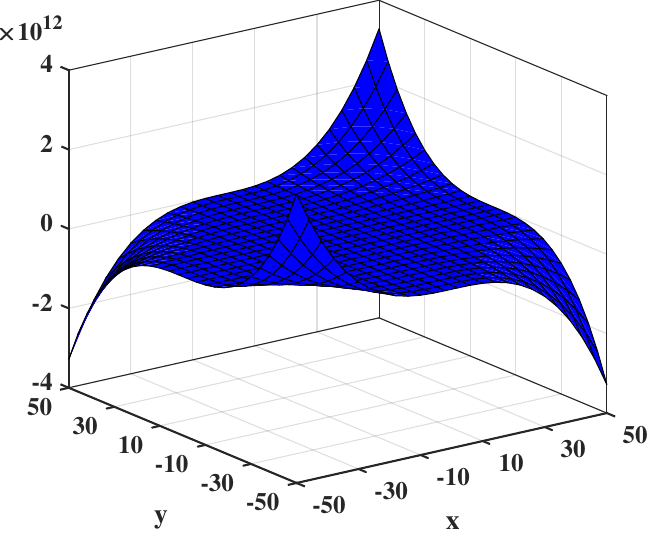}
		\centerline{\textbf{(a)}}
		\label{Fig:MP_PAH}
	\end{minipage}
	\hspace{0.2cm}
	\begin{minipage}[t]{0.45\textwidth}
		\includegraphics[width=\textwidth]{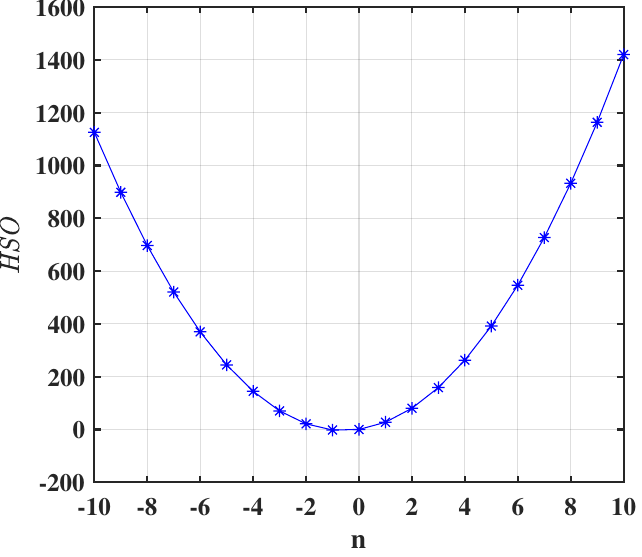}
		\centerline{\textbf{(b)}}
		\label{Fig:HSO_PAH}
	\end{minipage}
	\caption{Graphical illustration of the (a)~M-polynomial of $\textit{PAH}_n$ for $n=5$ and (b)~\textit{HSO} index of $\textit{PAH}_n$.}
\end{figure}
\subsection{Hyperbolic Sombor Index for V-Phynelenic Nanotubes\\$\boldsymbol{\textit{VPHX}[m,n]}$ and V-Phynelenic Nanotori\\$\boldsymbol{\textit{VPHY}[m,n]}$}
The most extensively researched nanostructures are V-Phenylenic nanotubes and nanotorus because of their extensive use in the manufacturing of corrosion-resistant, gas-sensing, and catalytic materials. A trivalent decoration created by alternating $C_4$, $C_6$ and $C_8$ is called a $C_4C_6C_8$ net. Several $C_4C_6C_8$ nets make up the architecture of V-Phenylenic nanotubes and nanotorus. The polycyclic conjugated molecules known as phenylenes are made up of both four- and six-membered rings. Each four-membered ring is next to two six-membered rings, and each four-membered ring is next to two eight-membered rings. $\textit{VPHX}[m,n]$ and $\textit{VPHY}[m,n]$, respectively, are V-Phenylenic nanotubes and nanotori, where $m$ and $n$ are the number of atoms in rows and columns, respectively~\cite{kwun2017}.

Let $G=\textit{VPHX}[m,n]$ represent a family of V-Phenylenic nanotubes with $6mn$ vertices and $9mn$ edges. Next, we use the M-polynomial-based formula given in Theorem~\ref{Th:28} to find the value of the hyperbolic Sombor index for $\textit{VPHX}[m,n]$.
\begin{theorem}~\cite{kwun2017}
	\label{Th:28}
	Let $G$ be the family of V-Phynelenic nanotubes $\textit{VPHX}[m,n]$, then the expression of M-polynomial is
	\begin{equation*}
		M(G;x,y)=4mx^2y^3+m(9n-5)x^3y^3.
	\end{equation*}
\end{theorem}
\begin{theorem}
	Let $G$ be the family of V-Phynelenic nanotubes $\textit{VPHX}[m,n]$. Then the hyperbolic Sombor index is given by
	$$\textit{HSO}(G)=2\sqrt{13}m+\sqrt{2}m(9n-5).$$
\end{theorem}
\begin{proof}
	The M-polynomial of $\textit{VPHX}[m,n]$ is given by
	\begin{equation*}
		M(G;x,y)=4mx^2y^3+m(9n-5)x^3y^3.
	\end{equation*}
	Then
	\begin{align*}
		\textit{HSO}(G)&= D^{1/2}_{x}JP_{y}P_{x}S_{x}(M(G;x,y))|_{x=1}\\
		&=D^{1/2}_{x}JP_{y}P_{x}S_{x}\Big(4mx^2y^3+m(9n-5)x^3y^3\Big)|_{x=1}\\
		&=D^{1/2}_{x}JP_{y}P_{x}\Big(2mx^2y^3+\Big(\frac{m(9n-5)}{3}\Big)m(9n-5)x^3y^3\Big)|_{x=1}\\
		&=D^{1/2}_{x}J\Big(2mx^4y^9+\Big(\frac{m(9n-5)}{3}\Big)m(9n-5)x^9y^9\Big)|_{x=1}\\
		&=D^{1/2}_{x}\Big(2mx^{13}+\Big(\frac{m(9n-5)}{3}\Big)m(9n-5)x^{18}\Big)|_{x=1}\\
		&=\Big(2\sqrt{13}mx^{13}+\sqrt{18}\times\Big(\frac{m(9n-5)}{3}\Big)m(9n-5)x^{18}\Big)|_{x=1}\\
		&=2\sqrt{13}m+\sqrt{2}m(9n-5).
	\end{align*}
\end{proof}
\begin{figure}[htb!]
	\centering
	\begin{minipage}[t]{0.45\textwidth}
		\includegraphics[width=\textwidth]{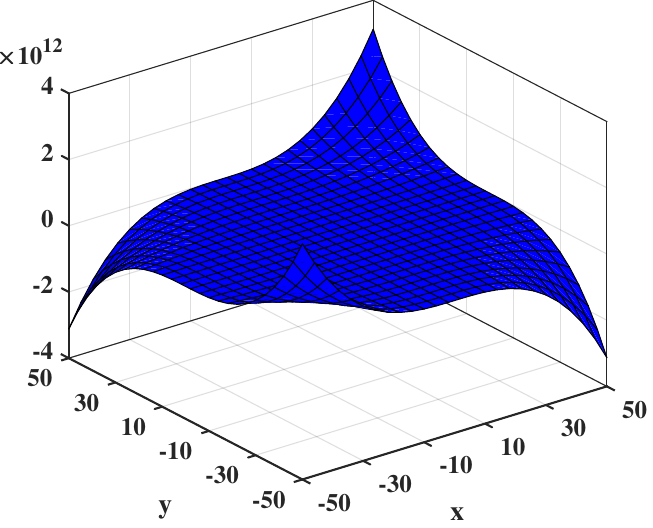}
		\centerline{\textbf{(a)}}
		\label{Fig:MP_VPHX}
	\end{minipage}
	\hspace{0.2cm}
	\begin{minipage}[t]{0.45\textwidth}
		\includegraphics[width=\textwidth]{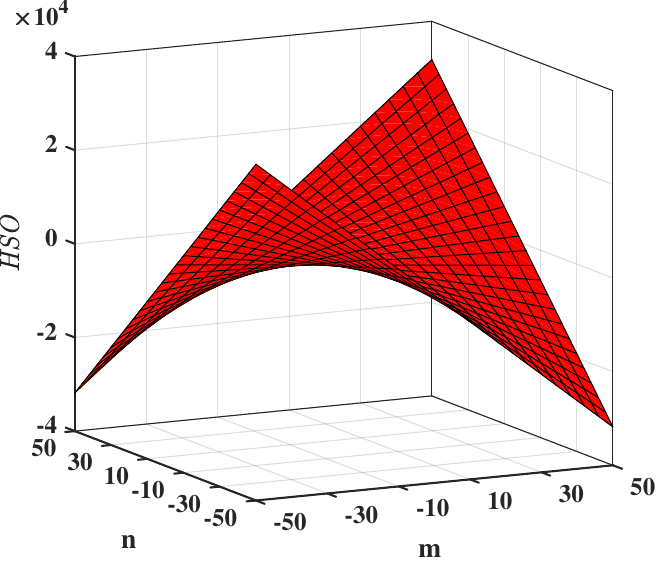}
		\centerline{\textbf{(b)}}
		\label{Fig:HSO_VPHX}
	\end{minipage}
	\caption{Graphical illustration of the (a)~M-polynomial of $\textit{VPHX}[m,n]$ for $m=n=5$ and (b)~\textit{HSO} index of $\textit{VPHX}[m,n]$.}
\end{figure}
Let $G=\textit{VPHY}[m,n]$ represent a family of V-Phenylenic nanotori with $6mn$ vertices and $9mn$ edges. Next, we use the M-polynomial-based formula given in Theorem~\ref{Th:30} to find the value of the hyperbolic Sombor index for $\textit{VPHY}[m,n]$.
\begin{theorem}~\cite{kwun2017}
	\label{Th:30}
	Let $G$ be the family of V-Phynelenic nanotori $\textit{VPHY}[m,n]$, then the expression of M-polynomial is
	\begin{equation*}
		M(G;x,y)=9mnx^3y^3.
	\end{equation*}
\end{theorem}
\begin{theorem}
	Let $G$ be the family of V-Phynelenic nanotori $\textit{VPHY}[m,n]$. Then the hyperbolic Sombor index is given by
	$$\textit{HSO}(G)=3\sqrt{18}mn.$$
\end{theorem}
\begin{proof}
	The M-polynomial of $\textit{VPHY}[m,n]$ is given by
	\begin{equation*}
		M(G;x,y)=9mnx^3y^3.
	\end{equation*}
	Then
	\begin{align*}
		\textit{HSO}(G)&= D^{1/2}_{x}JP_{y}P_{x}S_{x}(M(G;x,y))|_{x=1}\\
		&=D^{1/2}_{x}JP_{y}P_{x}S_{x}\Big(9mnx^3y^3\Big)|_{x=1}\\
		&=D^{1/2}_{x}JP_{y}P_{x}\Big(3mnx^3y^3\Big)|_{x=1}\\
		&=D^{1/2}_{x}J\Big(3mnx^9y^9\Big)|_{x=1}\\
		&=D^{1/2}_{x}\Big(3mnx^{18}\Big)|_{x=1}\\
		&=\Big(3\sqrt{18}mnx^{18}\Big)|_{x=1}\\
		&=3\sqrt{18}mn.
	\end{align*}
\end{proof}
\begin{figure}[htb!]
	\centering
	\begin{minipage}[t]{0.45\textwidth}
		\includegraphics[width=\textwidth]{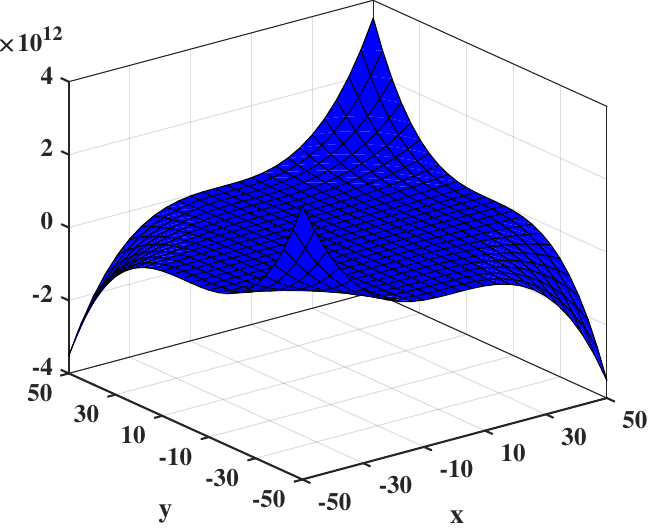}
		\centerline{\textbf{(a)}}
		\label{Fig:MP_VPHY}
	\end{minipage}
	\hspace{0.2cm}
	\begin{minipage}[t]{0.45\textwidth}
		\includegraphics[width=\textwidth]{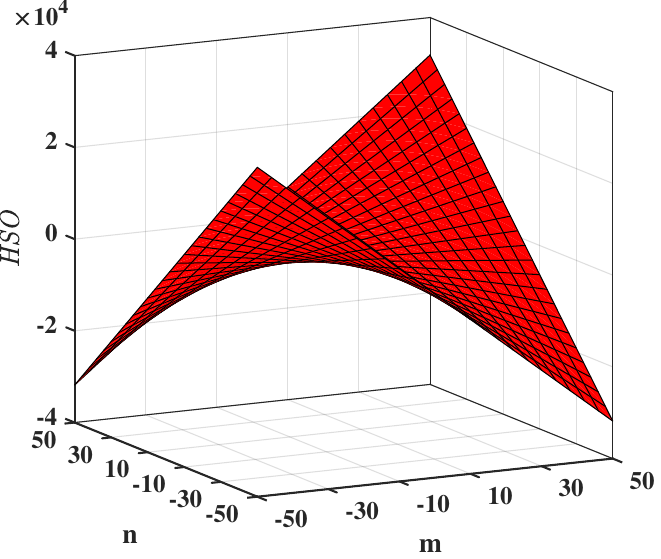}
		\centerline{\textbf{(b)}}
		\label{Fig:HSO_VPHY}
	\end{minipage}
	\caption{Graphical illustration of the (a)~M-polynomial of $\textit{VPHY}[m,n]$ for $m=n=5$ and (b)~\textit{HSO} index of $\textit{VPHY}[m,n]$.}
\end{figure}
\subsection{Hyperbolic Sombor Index for Porous Graphene~$\boldsymbol{\textit{PG}[p,q]}$}
Graphene has gained prominence as a groundbreaking material, inspiring extensive experimental and theoretical research due to its remarkable properties. Porous graphene, a class of graphene-based materials featuring nanoscale pores, exhibits exceptional properties across a wide range of applications. There are $12pq + 12p + 12q$ vertices and $15pq + 14p + 14q - 1$ edges in the molecular graph of porous graphene $G=\textit{PG}[p,q]$~\cite{shanmukha2021m}. We now determine the value of the hyperbolic Sombor index for $\textit{PG}[p,q]$ using the M-polynomial-based formula provided in Theorem~\ref{Th:32}.
\begin{theorem}~\cite{shanmukha2021m}
	\label{Th:32}
	Let $G$ be the family of porous graphene $\textit{PG}[p,q]$, then the expression of M-polynomial is
	\begin{equation*}
		M(G;x,y)=(4p+4q+4)x^2y^2+(12pq+8p+8q-4)x^2y^3+(3pq+2p+2q-1)x^3y^3.
	\end{equation*}
\end{theorem}
\begin{theorem}
	Let $G$ be the family of porous graphene $\textit{PG}[p,q]$. Then the hyperbolic Sombor index is given by
	$$\textit{HSO}(G)=\sqrt{8}(2p+2q+2)+\sqrt{13}(6pq+4p+4q-2)+\sqrt{2}(3pq+2p+2q-1).$$
\end{theorem}
\begin{proof}
	The M-polynomial of $\textit{PG}[p,q]$ is given by
	\begin{equation*}
		M(G;x,y)=(4p+4q+4)x^2y^2+(12pq+8p+8q-4)x^2y^3+(3pq+2p+2q-1)x^3y^3.
	\end{equation*}
	Then
	\begin{align*}
		&\textit{HSO}(G)\\
		&= D^{1/2}_{x}JP_{y}P_{x}S_{x}(M(G;x,y))|_{x=1}\\
		&=D^{1/2}_{x}JP_{y}P_{x}S_{x}\Big((4p+4q+4)x^2y^2+(12pq+8p+8q-4)x^2y^3+(3pq+2p+2q-1)x^3y^3\Big)|_{x=1}\\
		&=D^{1/2}_{x}JP_{y}P_{x}\Big((2p+2q+2)x^2y^2+(6pq+4p+4q-2)x^2y^3+\Big(\frac{3pq+2p+2q-1}{3}\Big)x^3y^3\Big)|_{x=1}\\
		&=D^{1/2}_{x}J\Big((2p+2q+2)x^4y^4+(6pq+4p+4q-2)x^4y^9+\Big(\frac{3pq+2p+2q-1}{3}\Big)x^9y^9\Big)|_{x=1}\\
		&=D^{1/2}_{x}\Big((2p+2q+2)x^8+(6pq+4p+4q-2)x^{13}+\Big(\frac{3pq+2p+2q-1}{3}\Big)x^{18}\Big)|_{x=1}\\
		&=\Big(\sqrt{8}(2p+2q+2)x^8+\sqrt{13}(6pq+4p+4q-2)x^{13}+\sqrt{18}\Big(\frac{3pq+2p+2q-1}{3}\Big)x^{18}\Big)|_{x=1}\\
		&=\sqrt{8}(2p+2q+2)+\sqrt{13}(6pq+4p+4q-2)+\sqrt{2}(3pq+2p+2q-1).
	\end{align*}
\end{proof}
\begin{figure}[htb!]
	\centering
	\begin{minipage}[t]{0.45\textwidth}
		\includegraphics[width=\textwidth]{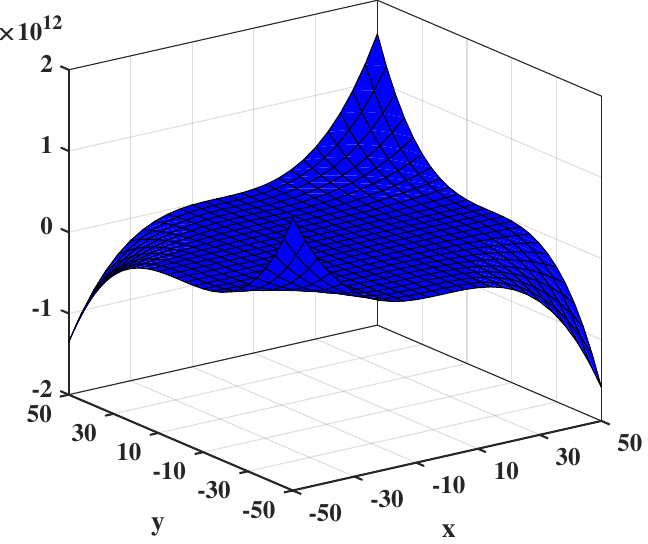}
		\centerline{\textbf{(a)}}
		\label{Fig:MP_PG}
	\end{minipage}
	\hspace{0.2cm}
	\begin{minipage}[t]{0.45\textwidth}
		\includegraphics[width=\textwidth]{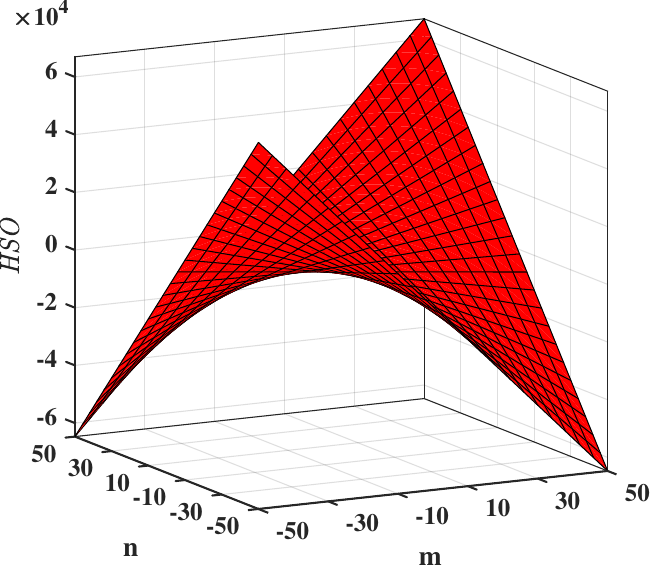}
		\centerline{\textbf{(b)}}
		\label{Fig:HSO_PG}
	\end{minipage}
	\caption{Graphical illustration of the (a)~M-polynomial of $\textit{PG}[p,q]$ for $p=q=5$ and (b)~\textit{HSO} index of $\textit{PG}[p,q]$.}
\end{figure}
\subsection{Hyperbolic Sombor Index for Tadpole Graph~$\boldsymbol{\textit{T}(n,m)}$}
A structure created by joining six segments of linear polyacenes into a closed loop is known as the tadpole graph $\textit{T}(n,m)$. It comprises a path graph on $m$ vertices and a cycle graph on $n(\ge 3)$ vertices, joined by a bridge~\cite{demaio2014fibonacci}. Next, we use the M-polynomial-based formula given in Theorem~\ref{Th:34} to get the value of the hyperbolic Sombor index for $\textit{T}(n,m)$.
\begin{theorem}~\cite{chaudhry2021m}
	\label{Th:34}
	Let $G$ be the family of tadpole graph $\textit{T}(n,m)$, then the expression of M-polynomial is
	\begin{equation*}
		M(G;x,y)=xy^2+(n+m-4)x^2y^2+3x^2y^3.
	\end{equation*}
\end{theorem}
\begin{theorem}
	Let $G$ be the family of tadpole graph $\textit{T}(n,m)$. Then the hyperbolic Sombor index is given by
	$$\textit{HSO}(G)=\sqrt{5}+\sqrt{2}(n+m-4)+\frac{3\sqrt{13}}{2}.$$
\end{theorem}
\begin{proof}
	The M-polynomial of $\textit{T}(n,m)$ is given by
	\begin{equation*}
		M(G;x,y)=xy^2+(n+m-4)x^2y^2+3x^2y^3.
	\end{equation*}
	Then
	\begin{align*}
		\textit{HSO}(G)&= D^{1/2}_{x}JP_{y}P_{x}S_{x}(M(G;x,y))|_{x=1}\\
		&=D^{1/2}_{x}JP_{y}P_{x}S_{x}\Big(xy^2+(n+m-4)x^2y^2+3x^2y^3\Big)|_{x=1}\\
		&=D^{1/2}_{x}JP_{y}P_{x}\Big(xy^2+\Big(\frac{n+m-4}{2}\Big)x^2y^2+\frac{3}{2}x^2y^3\Big)|_{x=1}\\
		&=D^{1/2}_{x}J\Big(xy^4+\Big(\frac{n+m-4}{2}\Big)x^4y^4+\frac{3}{2}x^4y^9\Big)|_{x=1}\\
		&=D^{1/2}_{x}\Big(x^5+\Big(\frac{n+m-4}{2}\Big)x^8+\frac{3}{2}x^{13}\Big)|_{x=1}\\
		&=\Big(\sqrt{5}x^5+\sqrt{8}\Big(\frac{n+m-4}{2}\Big)x^8+\frac{3\sqrt{13}}{2}x^{13}\Big)|_{x=1}\\
		&=\sqrt{5}+\sqrt{2}(n+m-4)+\frac{3\sqrt{13}}{2}.
	\end{align*}
\end{proof}
\begin{figure}[htb!]
	\centering
	\begin{minipage}[t]{0.45\textwidth}
		\includegraphics[width=\textwidth]{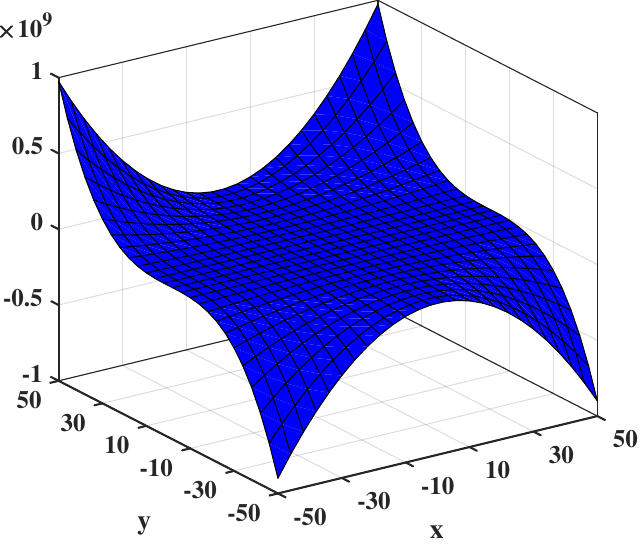}
		\centerline{\textbf{(a)}}
		\label{Fig:MP_Tad}
	\end{minipage}
	\hspace{0.2cm}
	\begin{minipage}[t]{0.45\textwidth}
		\includegraphics[width=\textwidth]{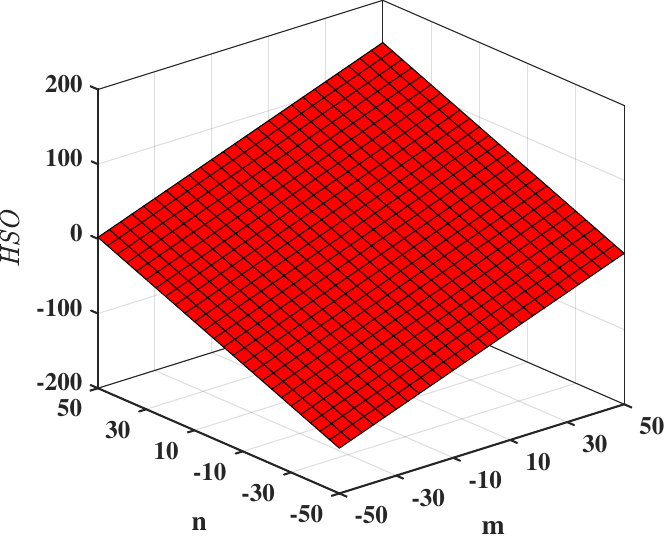}
		\centerline{\textbf{(b)}}
		\label{Fig:HSO_Tad}
	\end{minipage}
	\caption{Graphical illustration of the (a)~M-polynomial of $\textit{T}(n,m)$ for $n=m=5$ and (b)~\textit{HSO} index of $\textit{T}(n,m)$.}
\end{figure}
\subsection{Hyperbolic Sombor Index for Polyphenylenes~$\boldsymbol{\textit{P}[s,t]}$}
The nuclei of benzenoid aromatics, which are joined by carbon-carbon bonds, are found in polyphenylenes. The chemical graph $\textit{P}[s,t]$ is made up of $s$ rows and $t$ columns. It contains $24st + 12s$ vertices and $30st + 13s - t$ edges~\cite{zuo2020}. With the aid of the M-polynomial provided in Theorem~\ref{Th:36}, we determine the values of the hyperbolic Sombor index for polyphenylenes.
\begin{theorem}~\cite{zuo2020}
	\label{Th:36}
	Let $G$ be the family of polyphenylenes $\textit{P}[s,t]$, then the expression of M-polynomial is
	\begin{equation*}
		M(G;x,y)=4(2s+t)x^2y^2+4(6st+s-t)x^2y^3+(6st+s-t)x^3y^3.
	\end{equation*}
\end{theorem}
\begin{theorem}
	Let $G$ be the family of polyphenylenes $\textit{P}[s,t]$. Then the hyperbolic Sombor index is given by
	$$\textit{HSO}(G)=4\sqrt{2}(2s+t)+(2\sqrt{13}+\sqrt{2})(6st+s-t).$$
\end{theorem}
\begin{proof}
	The M-polynomial of $\textit{P}[s,t]$ is given by
	\begin{equation*}
		M(G;x,y)=4(2s+t)x^2y^2+4(6st+s-t)x^2y^3+(6st+s-t)x^3y^3.
	\end{equation*}
	Then
	\begin{align*}
		\textit{HSO}(G)&= D^{1/2}_{x}JP_{y}P_{x}S_{x}(M(G;x,y))|_{x=1}\\
		&=D^{1/2}_{x}JP_{y}P_{x}S_{x}\Big(4(2s+t)x^2y^2+4(6st+s-t)x^2y^3+(6st+s-t)x^3y^3\Big)|_{x=1}\\
		&=D^{1/2}_{x}JP_{y}P_{x}\Big(2(2s+t)x^2y^2+2(6st+s-t)x^2y^3+\Big(\frac{6st+s-t}{3}\Big)x^3y^3\Big)|_{x=1}\\
		&=D^{1/2}_{x}J\Big(2(2s+t)x^4y^4+2(6st+s-t)x^4y^9+\Big(\frac{6st+s-t}{3}\Big)x^9y^9\Big)|_{x=1}\\
		&=D^{1/2}_{x}\Big(2(2s+t)x^8+2(6st+s-t)x^{13}+\Big(\frac{6st+s-t}{3}\Big)x^{18}\Big)|_{x=1}\\
		&=\Big(2\sqrt{8}(2s+t)x^8+2\sqrt{13}(6st+s-t)x^{13}+\sqrt{18}\Big(\frac{6st+s-t}{3}\Big)x^{18}\Big)|_{x=1}\\
		&=4\sqrt{2}(2s+t)+(2\sqrt{13}+\sqrt{2})(6st+s-t).
	\end{align*}
\end{proof}
\begin{figure}[htb!]
	\centering
	\begin{minipage}[t]{0.45\textwidth}
		\includegraphics[width=\textwidth]{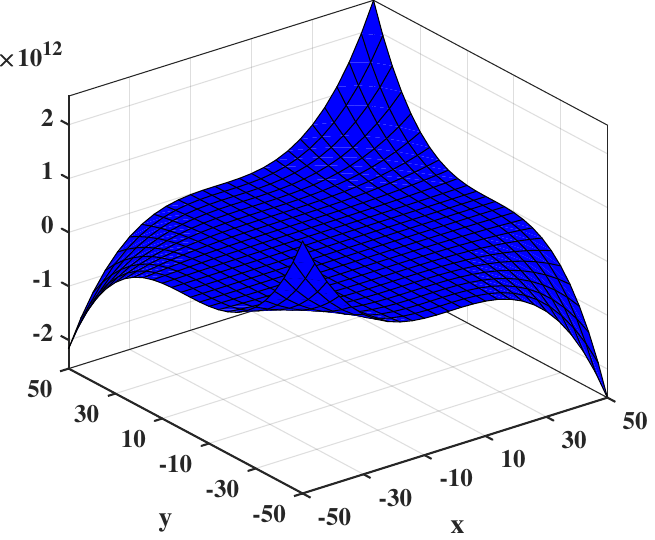}
		\centerline{\textbf{(a)}}
		\label{Fig:MP_Poly}
	\end{minipage}
	\hspace{0.2cm}
	\begin{minipage}[t]{0.45\textwidth}
		\includegraphics[width=\textwidth]{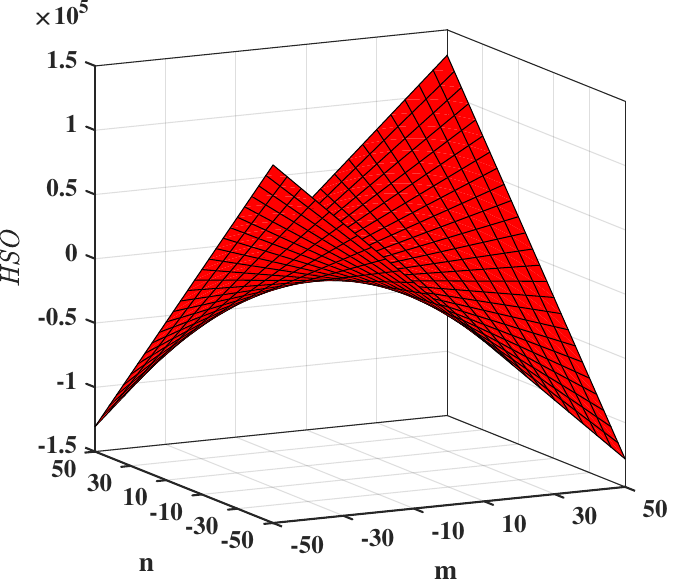}
		\centerline{\textbf{(b)}}
		\label{Fig:HSO_Poly}
	\end{minipage}
	\caption{Graphical illustration of the (a)~M-polynomial of $\textit{P}[s,t]$ for $s=t=5$ and (b)~\textit{HSO} index of $\textit{P}[s,t]$.}
\end{figure}
We also demonstrate the numerical calculation of the \textit{HSO} index for each of the families of graphs under consideration using the MATLAB R2019 software.
\begin{table}[htb!]
	\caption{Values of the hyperbolic Sombor index for various chemical families with two variables.}
	\label{Table:Cor2}
	\renewcommand*{\arraystretch}{1.2}
	\tiny
	\resizebox{\textwidth}{!}
	{
		\begin{tabular}{|c|c|c|c|c|c|c|c|c|c|c|}
			\hline
			 & $\textbf{[1,1]}$ & $\textbf{[2,2]}$ & $\textbf{[3,3]}$ & $\textbf{[4,4]}$ & $\textbf{[5,5]}$ & $\textbf{[6,6]}$ & $\textbf{[7,7]}$& $\textbf{[8,8]}$ & $\textbf{[9,9]}$ & $\textbf{[10,10]}$\\ \hline
			$\boldsymbol{\mathcal{B}_{\alpha}(a,b)}$ & $233.6269$ &	$1256.1506$ &	$3068.6253$ &	$5671.0512$ &	$9063.4283$ &	$13245.7564$ &	$18218.0357$ &	$23980.2661$ &	$30532.4477$ &	$37874.5803$\\ \hline
			$\boldsymbol{B_{m,n}}$ & $12.8680$ &	$49.9176$ &	$103.9378$ &	$174.9285$ &	$262.8898$ &	$367.8217$ &	$489.7241$ &	$628.5971$ &	$784.4407$ &	$957.2548$\\ \hline
			$\boldsymbol{\textit{\textbf{VPHX}}[m,n]}$ & $12.8680$ &	$49.9176$ &	$103.9378$ &	$174.9285$ &	$262.8898$ &	$367.8217$ &	$489.7241$ &	$628.5971$ &	$784.4407$ &	$957.2548$\\ \hline
			$\boldsymbol{\textit{\textbf{VPHY}}[m,n]}$ & $12.7279$ &	$50.9117$ &	$114.5513$ &	$203.6468$ &	$318.1981$ &	$458.2052$ &	$623.6682$ &	$814.5870$ &	$1030.9617$ &	$1272.7922$\\ \hline
			$\boldsymbol{\textit{\textbf{PG}}[p,q]}$ & $68.7225$ &	$192.1653$ &	$367.3600$ &	$594.3066$ &	$873.0051$ &	$1203.4555$ &	$1585.6578$ &	$2019.6120$ &	$2505.3181$ &	$3042.7761$\\ \hline
			$\boldsymbol{\textit{\textbf{T}}(n,m)}$ & $4.8160$ &	$7.6444$ &	$10.4728$ &	$13.3012$ &	$16.1297$ &	$18.9581$ &	$21.7865$ &	$24.6150$ &	$27.4434$ &	$30.2718$\\ \hline
			$\boldsymbol{\textit{\textbf{P}}[s,t]}$ & $68.7225$ &	$240.9487$ &	$516.6788$ &	$895.9126$ &	$1378.6502$ &	$1964.8917$ &	$2654.6369$ &	$3447.8859$ &	$4344.6387$ &	$5344.8953$\\ \hline
		\end{tabular}
	}
\end{table}
\begin{table}[htb!]
	\caption{Values of the hyperbolic Sombor index for various chemical families with one variable.}
	\label{Table:4}
	\centering
	\renewcommand*{\arraystretch}{1.5}
	\tiny
	\resizebox{\textwidth}{!}
	{
		\begin{tabular}{|c|c|c|c|c|c|c|c|c|c|c|}
			\hline
			$\textbf{n}$ & $\textbf{1}$ & $\textbf{2}$ & $\textbf{3}$ & $\textbf{4}$ & $\textbf{5}$ & $\textbf{6}$& $\textbf{7}$ & $\textbf{8}$ & $\textbf{9}$ & $\textbf{10}$\\ \hline
			\textbf{\textit{PETIM}} & $35.0878$ &	$106.4481$ &	$249.1686$ &	$534.6097$ &	$1105.4919$ &	$2247.2564$ &	$4530.7853$ &	$9097.8430$ &	$18231.9586$ &	$36500.1896$\\ \hline
			$\boldsymbol{D_nP_n}$ & $219.7848$ &	$457.4574$ &	$695.1300$ &	$932.8025$ &	$1170.4751$ &	$1408.1477$ &	$1645.8202$ &	$1883.4928$ &	$2121.1654$ &	$2358.8379$\\ \hline
			$\boldsymbol{\textit{\textbf{DPZ}}_n}$ & $243.6667$ &	$501.0258$ &	$1015.7441$ &	$2045.1806$ &	$4104.0537$ &	$8221.7999$ &	$16457.2922$ &	$32928.2769$ &	$65870.2464$ &	$131754.1852$\\ \hline
			\textbf{\textit{PETAA}} & $126.6894$ &	$287.2420$ &	$608.3473$ &	$1250.5578$ &	$2534.9787$ &	$5103.8207$ &	$10241.5047$ &	$20516.8727$ &	$41067.6087$ &	$82169.0806$\\ \hline
			$\boldsymbol{\textit{\textbf{PAH}}_n}$ & $27.4589$ &	$80.3737$ &	$158.7444$ &	$262.5709$ &	$391.8532$ &	$546.5913$ &	$726.7854$ &	$932.4352$ &	$1163.5409$ &	$1420.1025$\\ \hline
		\end{tabular}
	}
\end{table}
\section{Conclusion}
\label{Sec:Con}
We suggested the closed derivation formula in this work to calculate the recently introduced hyperbolic Sombor index (\textit{HSO}) of a graph, which is written in terms of a graph's M-polynomial. Additionally, we used our closed derivation formula to evaluate this index for some standard families of graphs, chemical families of $\mathcal{B}_{\alpha}(a,b)$ and dendrimers, jagged-rectangle benzenoid system~$B_{m,n}$, polycyclic aromatic hydrocarbons~$\textit{PAH}_{n}$, V-Phynelenic nanotubes~$\textit{VPHX}[m,n]$, V-Phynelenic nanotori~$\textit{VPHY}[m,n]$, porous graphene $\textit{PG}[p,q]$, tadpole graph~$T(n,m)$ and polyphenylenes~$P[s,t]$, by using their respective M-polynomials. Moreover, numerical and graphical representations are provided for the M-polynomial and the computed \textit{HSO} index of the chemical families for different ranges. The values of the \textit{HSO} index have increased as we increase the values of $m$ and $n$ for all of the chemical families. Since topological indices can be directly used in \textit{QSPR}/\textit{QSAR} studies of chemical substances, researchers in the field may find benefits in the hyperbolic Sombor index results for different graph families.
\section*{Acknowledgment}
The first author (Jayjit Barman) is grateful to the UNIVERSITY GRANTS COMMISSION, Ministry of Social Justice and Empowerment, New Delhi, India for awarding the National Fellowship for Scheduled Caste Students (NFSC) with reference to UGC-Ref. No.: 211610011412/(CSIR-UGC NET JUNE 2021) dated 11-May-2022. The second author (Dr.\ Shibsankar Das) is obliged to the Development Cell, Banaras Hindu University for financially supporting this work through the Faculty ``Incentive Grant'' under the Institute of Eminence (IoE) Scheme for the year 2024-25 (Project sanction order number: R/Dev/D/IoE/Incentive (Phase-IV)/2024-25/82483, dated 7 January 2025).
\section*{Author Contributions}
All the authors contributed equally.
\section*{Data Availability Statement}
The data used to support the findings of this study are included within the article.
\section*{Disclosure Statement \& Competing Interests}
No disclosure statement and potential conflict of interest was reported by the authors.


\begin{thebibliography}{10}
	\expandafter\ifx\csname url\endcsname\relax
	\def\url#1{\texttt{#1}}\fi
	\expandafter\ifx\csname urlprefix\endcsname\relax\def\urlprefix{URL }\fi
	\expandafter\ifx\csname href\endcsname\relax
	\def\href#1#2{#2} \def\path#1{#1}\fi
	
	\bibitem{trinajstic1992}
	N.~Trinajsti{\'c}, Chemical graph theory, 2nd Edition, Mathematical Chemistry
	Series, CRC Press, 1992.
	
	\bibitem{west2000}
	D.~B. West, Introduction to graph theory, 2nd Edition, Prentice hall, 2000.
	
	\bibitem{barman2025}
	J.~Barman, S.~Das, Geometric \text{A}pproach to \text{D}egree-\text{B}ased
	\text{T}opological \text{I}ndex: \text{H}yperbolic \text{S}ombor
	\text{I}ndex, MATCH Commun. Math. Comput. Chem. 95~(1) (2026) 63--94.
	\newblock \href {https://doi.org/10.46793/match95-1.03425}
	{\path{doi:10.46793/match95-1.03425}}.
	
	\bibitem{sdas2022b}
	S.~Das, V.~Kumar,
	\href{https://ijmc.kashanu.ac.ir/article_112048.html}{Investigation of closed
		derivation formulas for \text{GQ} and \text{QG} indices of a graph via
		\text{M}-polynomial}, Iranian Journal of Mathematical Chemistry 13~(2) (2022)
	129--144.
	\newblock \href {https://doi.org/10.22052/ijmc.2022.246172.1614}
	{\path{doi:10.22052/ijmc.2022.246172.1614}}.
	\newline\urlprefix\url{https://ijmc.kashanu.ac.ir/article_112048.html}
	
	\bibitem{sdas2022g}
	S.~Das, S.~Rai, On derivation formulas of the \text{N}irmala indices from the
	\text{M}-polynomial of a graph, preseanted in International E-conference on
	Pure and Applied Mathematical Sciences (ICPAMS-2022) during 04-07 May 2022
	(2022).
	
	\bibitem{deutsch2015}
	E.~Deutsch, S.~Klav{\v{z}}ar, \text{M-polynomial} and degree-based topological
	indices, Iranian Journal of Mathematical Chemistry 6~(2) (2015) 93--102.
	\newblock \href {https://doi.org/10.22052/ijmc.2015.10106}
	{\path{doi:10.22052/ijmc.2015.10106}}.
	
	\bibitem{hosoya1988}
	H.~Hosoya, On some counting polynomials in chemistry, Discrete Applied
	Mathematics 19~(1-3) (1988) 239--257.
	\newblock \href {https://doi.org/10.1016/0166-218X(88)90017-0}
	{\path{doi:10.1016/0166-218X(88)90017-0}}.
	
	\bibitem{kauffman1989}
	L.~H. Kauffman, A tutte polynomial for signed graphs, Discrete Applied
	Mathematics 25~(1-2) (1989) 105--127.
	
	\bibitem{zhang1996}
	H.~Zhang, F.~Zhang, The clar covering polynomial of hexagonal systems 1,
	Discrete Applied Mathematics 69~(1-2) (1996) 147--167.
	
	\bibitem{farrell1979}
	E.~J. Farrell, An introduction to matching polynomials, Journal of
	Combinatorial Theory, Series B 27~(1) (1979) 75--86.
	
	\bibitem{hassani2013}
	F.~Hassani, A.~Iranmanesh, S.~Mirzaie, Schultz and modified schultz polynomials
	of $\text{C}100$ fullerene, MATCH Commun. Math. Comput. Chem 69 (2013)
	87--92.
	
	\bibitem{mondal2021b}
	S.~Mondal, M.~K. Siddiqui, N.~De, A.~Pal, Neighborhood \text{M}-polynomial of
	crystallographic structures, Biointerface Res. Appl. Chem. 11 (2021)
	9372--9381.
	
	\bibitem{sdas2023}
	S.~Das, S.~Rai, V.~Kumar, On topological indices of molnupiravir and its qspr
	modelling with some other antiviral drugs to treat covid-19 patients, Journal
	of Mathematical Chemistry (2023) 1--44\href
	{https://doi.org/10.1007/s10910-023-01518-z}
	{\path{doi:10.1007/s10910-023-01518-z}}.
	
	\bibitem{sdas2020a}
	S.~Das, S.~Rai, \text{M}-polynomial and related degree-based topological
	indices of the third type of \text{H}ex-derived network, Nanosystems:
	Physics, Chemistry, Mathematics 11~(3) (2020) 267--274.
	\newblock \href {https://doi.org/10.17586/2220-8054-2020-11-3-267-274}
	{\path{doi:10.17586/2220-8054-2020-11-3-267-274}}.
	
	\bibitem{sdas2022a}
	S.~Das, V.~Kumar,
	\href{https://pjm.ppu.edu/paper/1035-m-polynomial-two-dimensional-silicon-carbons}{On
		\text{M}-polynomial of the \text{T}wo-dimensional
		\text{S}ilicon-\text{C}arbons}, Palestine Journal of Mathematics 11~(Special
	Issue II) (2022) 136--157.
	\newline\urlprefix\url{https://pjm.ppu.edu/paper/1035-m-polynomial-two-dimensional-silicon-carbons}
	
	\bibitem{sdas2022e}
	S.~Das, S.~Rai, On \text{M-}polynomial and associated topological descriptors
	of subdivided hex-derived network of type three., Journal of Computational
	Technologies (2022. in press).
	
	\bibitem{deng2011}
	H.~Deng, J.~Yang, F.~Xia, A general modeling of some vertex-degree based
	topological indices in benzenoid systems and phenylenes, Computers \&
	Mathematics with Applications 61~(10) (2011) 3017--3023.
	\newblock \href {https://doi.org/10.1016/j.camwa.2011.03.089}
	{\path{doi:10.1016/j.camwa.2011.03.089}}.
	
	\bibitem{afzal2020}
	F.~Afzal, S.~Hussain, D.~Afzal, S.~Razaq, Some new degree based topological
	indices via \text{M}-polynomial, Journal of Information and Optimization
	Sciences 41~(4) (2020) 1061--1076.
	
	\bibitem{basavanagoud2019}
	B.~Basavanagoud, A.~P. Barangi, P.~Jakkannavar, M-polynomial of some graph
	operations and cycle related graphs, Iranian Journal of Mathematical
	Chemistry 10~(2) (2019) 127--150.
	
	\bibitem{basavanagoud2020}
	B.~Basavanagoud, P.~Jakkannavar, M-polynomial and degree-based topological
	indices of graphs, Electronic Journal of Mathematical Analysis and
	Applications 8~(1) (2020) 75--99.
	
	\bibitem{sdas2021}
	S.~Das, S.~Rai, Topological characterization of the third type of triangular
	\text{H}ex-derived networks, Scientific Annals of Computer Science 31~(2)
	(2021) 145--161.
	\newblock \href {https://doi.org/10.7561/SACS.2021.2.145}
	{\path{doi:10.7561/SACS.2021.2.145}}.
	
	\bibitem{sdas2022c}
	S.~Das, S.~Rai, Degree-based topological descriptors of type 3 rectangular
	hex-derived networks., Bulletin of the Institute of Combinatorics and its
	Applications (BICA) 95 (2022) 21--37.
	
	\bibitem{hussain2021}
	S.~Hussain, A.~Alsinai, D.~Afzal, A.~Maqbool, F.~Afzal, M.~Cancan,
	Investigation of closed formula and topological properties of remdesivir
	$(\text{C}_{27}\text{H}_{35}\text{N}_{6}\text{O}_{8}\text{P})$, Chemical
	Methodologies 5~(6) (2021) 485--497.
	
	\bibitem{xavier2024computing}
	D.~A. Xavier, T.~Nair~A, M.~Usman~Ghani, A.~Baby, F.~Tchier, Computing
	molecular descriptors of boron icosahedral sheet, International Journal of
	Quantum Chemistry 124~(13) (2024) e27443.
	\newblock \href {https://doi.org/10.1002/qua.27443}
	{\path{doi:10.1002/qua.27443}}.
	
	\bibitem{tomalia1985new}
	D.~A. Tomalia, H.~Baker, J.~Dewald, M.~Hall, G.~Kallos, S.~Martin, J.~Roeck,
	J.~Ryder, P.~Smith, A new class of polymers: starburst-dendritic
	macromolecules, Polymer journal 17~(1) (1985) 117--132.
	\newblock \href {https://doi.org/10.1295/polymj.17.117}
	{\path{doi:10.1295/polymj.17.117}}.
	
	\bibitem{sdas2023a}
	S.~Das, S.~Rai, On closed derivation formulas of \text{N}irmala indices from
	the \text{M-polynomial} of a graph, Journal of the Indian Chemical Society
	100~(6) (2023) 101017.
	\newblock \href {https://doi.org/10.1016/j.jics.2023.101017}
	{\path{doi:10.1016/j.jics.2023.101017}}.
	
	\bibitem{shui2005}
	S.~Ling-Ling, W.~Zhi-Ning, Z.~Li-Qiang, Rheological \text{P}roperties of
	\text{C}ubic \text{L}iquid \text{C}rystals \text{F}ormed from
	\text{M}onoglyceride/\text{H2O} \text{S}ystems, Chinese Journal of Chemistry
	23~(3) (2005) 245--250.
	\newblock \href {https://doi.org/10.1002/cjoc.200590245}
	{\path{doi:10.1002/cjoc.200590245}}.
	
	\bibitem{kwun2018}
	W.~N. M. A.~C. Young Chel~Kwun, Ashaq~Ali, S.~M. Kang, \text{M-polynomials} and
	degree-based topological indices of triangular, hourglass, and
	jagged-rectangle benzenoid systems, Journal of Chemistry 2018 (2018) 1--8.
	\newblock \href {https://doi.org/10.1155/2018/8213950}
	{\path{doi:10.1155/2018/8213950}}.
	
	\bibitem{shao2018computing}
	Z.~Shao, M.~K. Siddiqui, M.~H. Muhammad, Computing \text{Z}agreb indices and
	\text{Z}agreb polynomials for symmetrical nanotubes, Symmetry 10~(7) (2018)
	244.
	\newblock \href {https://doi.org/10.3390/sym10070244}
	{\path{doi:10.3390/sym10070244}}.
	
	\bibitem{jamil2016computing}
	M.~K. Jamil, M.~R. Farahani, M.~I. Shabir, M.~A. Malik, Computing
	\text{E}ccentric \text{V}ersion of \text{S}econd \text{Z}agreb \text{I}ndex
	of \text{P}olycyclic \text{A}romatic \text{H}ydrocarbons \text{PAHk}, Applied
	Mathematics and Nonlinear Sciences 1~(1) (2016) 247--252.
	\newblock \href {https://doi.org/10.21042/AMNS.2016.1.00019}
	{\path{doi:10.21042/AMNS.2016.1.00019}}.
	
	\bibitem{dietz20002}
	F.~Dietz, N.~Tyutyulkov, G.~Madjarova, K.~M{\"u}llen, Is \text{2-D} graphite an
	ultimate large hydrocarbon? \text{II.} \text{S}tructure and energy spectra of
	polycyclic aromatic hydrocarbons with defects, The Journal of Physical
	Chemistry B 104~(8) (2000) 1746--1761.
	\newblock \href {https://doi.org/10.1021/jp9928659}
	{\path{doi:10.1021/jp9928659}}.
	
	\bibitem{kwun2017}
	Y.~C. Kwun, M.~Munir, W.~Nazeer, S.~Rafique, S.~M. Kang, M-polynomials and
	topological indices of v-phenylenic nanotubes and nanotori, Scientific
	Reports 7~(1) (2017) 1--9.
	\newblock \href {https://doi.org/10.1038/s41598-017-08309-y}
	{\path{doi:10.1038/s41598-017-08309-y}}.
	
	\bibitem{shanmukha2021m}
	M.~Shanmukha, A.~Usha, K.~Shilpa, N.~Basavarajappa, M-polynomial and
	neighborhood \text{M}-polynomial methods for topological indices of porous
	graphene, The European physical journal plus 136 (2021) 1--16.
	\newblock \href {https://doi.org/10.1140/epjp/s13360-021-02074-8}
	{\path{doi:10.1140/epjp/s13360-021-02074-8}}.
	
	\bibitem{demaio2014fibonacci}
	J.~DeMaio, J.~Jacobson, Fibonacci number of the tadpole graph, Electronic
	Journal of Graph Theory and Applications (EJGTA) 2~(2) (2014) 129--138.
	\newblock \href {https://doi.org/10.5614/ejgta.2014.2.2.5}
	{\path{doi:10.5614/ejgta.2014.2.2.5}}.
	
	\bibitem{chaudhry2021m}
	F.~Chaudhry, M.~N. Husin, F.~Afzal, D.~Afzal, M.~Ehsan, M.~Cancan, M.~R.
	Farahani, M-polynomials and degree-based topological indices of tadpole
	graph, Journal of Discrete Mathematical Sciences and Cryptography 24~(7)
	(2021) 2059--2072.
	\newblock \href {https://doi.org/10.1080/09720529.2021.1984561}
	{\path{doi:10.1080/09720529.2021.1984561}}.
	
	\bibitem{zuo2020}
	X.~Zuo, M.~Numan, S.~I. Butt, M.~K. Siddiqui, R.~Ullah, U.~Ali, Computing
	topological indices for molecules structure of polyphenylene via
	m-polynomials, Polycyclic Aromatic Compounds (2020) 1--10\href
	{https://doi.org/10.1080/10406638.2020.1768413}
	{\path{doi:10.1080/10406638.2020.1768413}}.
	
\end{thebibliography}


\end{document}